\documentclass[reprint,longbibliography]{revtex4-1} 
\usepackage{amssymb,amsmath}
\usepackage{graphicx,import}
\usepackage{placeins}
\usepackage{hyperref}
\usepackage{color}
\usepackage{longtable}
\usepackage{amsthm}
\usepackage{algorithm}
\usepackage{algpseudocode} 
\usepackage{upgreek}
\usepackage[nameinlink,poorman]{cleveref}

\graphicspath{{./}} 
\newtheorem{theorem}{Theorem}

\newtheorem{corollary}[theorem]{Corollary}
\newtheorem{lemma}[theorem]{Lemma}
\newtheorem{proposition}[theorem]{Proposition}

\usepackage{enumitem}
\newlist{thmprop}{enumerate}{2}
\setlist[thmprop]{label={\normalfont(\roman*)},ref=(\roman*)}

\crefname{equation}{Eq.}{Eqs.}
\crefname{figure}{Fig.}{Figs.}
\crefname{table}{Table}{Tables} 
\crefname{section}{Section}{Sections}
\crefname{chapter}{Chapter}{Chapters}
\crefname{appendix}{Appendix}{Appendices}
\crefname{algorithm}{Algorithm}{Algorithms}

\crefname{theorem}{Theorem}{Theorems}
\crefname{defn}{Definiton}{Definitions}
\crefname{azm}{Assumption}{Assumptions}
\crefname{corollary}{Corollary}{Corollaries}
\crefname{lemma}{Lemma}{Lemmas}
\crefname{thmprop}{property}{properties}
\crefname{proposition}{Proposition}{Propositions}
\crefname{remark}{Remark}{Remarks}

\newcommand{\ex}[1]{\mathbb{E}[#1]} 
\newcommand{\measure}[2]{\mathbb{P}_{#1}{[#2]}}  
\newcommand{\ofieldgen}[1]{\sigma(#1)}  

\newcommand{\erf}[0]{\mathrm{erf}}
\newcommand{\ofield}[0]{$\sigma$-field} 

\newcommand{\rhoq}[1]{\rho_{#1}} 

\newcommand{\ket}[1]{| #1 \rangle} 
\newcommand{\bra}[1]{\langle #1 |} 

\begin{document}

\title{Adaptive filtering of projective quantum measurements using discrete stochastic methods}

\author{Riddhi Swaroop Gupta} 
\email{riddhi.sw@gmail.com}
\altaffiliation[Present address: ]{IBM Quantum, IBM Australia Research Labs, Southbank, Victoria 3006, Australia}
\affiliation{ARC Centre of Excellence for Engineered Quantum Systems, School of Physics, The University of Sydney, New South Wales 2006, Australia}

\author{Michael J. Biercuk}
\affiliation{ARC Centre of Excellence for Engineered Quantum Systems, School of Physics, The University of Sydney, New South Wales 2006, Australia}

\begin{abstract}  
Adaptive filtering is a powerful class of control theoretic concepts useful in extracting information from noisy data sets or performing forward prediction in time for a dynamic system.  The broad utilization of the associated algorithms makes them attractive targets for similar problems in the quantum domain. To date, however, the construction of adaptive filters for quantum systems has typically been carried out in terms of stochastic differential equations for weak, continuous quantum measurements, as used in linear quantum systems such as optical cavities. Discretized measurement models are not as easily treated in this framework, but are frequently employed in quantum information systems leveraging projective measurements.  This paper presents a detailed analysis of several technical innovations that enable classical filtering of discrete projective measurements, useful for adaptively learning system-dynamics, noise properties, or hardware performance variations in classically correlated measurement data from quantum devices. In previous work we studied a specific case of this framework, in which noise and calibration errors on qubit arrays could be efficiently characterized in space; here, we present a generalized analysis of filtering in quantum systems and demonstrate that the traditional convergence properties of nonlinear classical filtering hold using  single-shot projective measurements. These results are important early demonstrations indicating that a range of concepts and techniques from classical nonlinear filtering theory may be applied to the characterization of quantum systems involving discretized projective measurements, paving the way for broader adoption of control theoretic techniques in quantum technology.
\end{abstract}

\maketitle
\section{Introduction}

Quantum computers in the NISQ-era face considerable challenges in mitigating the effects of decoherence on intermediate-scale multi-qubit devices. In realistic operating environments subject to noise, difficulties arise in device calibration, control and error mitigation and the complexity of these challenges increases with system size. As the number of qubits on a device increases, existing calibration and control techniques typically lead to an infeasible resource overhead at the expense of available compute time. In overcoming these contemporary challenges, insights from classical inference and control engineering literature generally appear to be relevant. Indeed, contemporary classical techniques for quantum systems characterization  \cite{GoogleSupremacy,lennon2019efficiently,tranter2018multiparameter}, adaptive tomography \cite{ferrie2014quantum,granade2015characterization,wiebe2015can,granade2015accelerated,granade2017practical}, and parameter estimation \cite{granade2012robust,stenberg2014efficient,kimmel2015robust,rudinger2017experimental} add to a growing body of literature in the last decade which has focused on realizing inexpensive and scalable characterization and control methods. Classical protocols have also been used for implementing optimal or efficient experiments \cite{huszar2012adaptive,kravtsov2013experimental} by enabling adaptive measurement selection or qubit allocation \cite{majumder2020real,gupta2020integration}. 

However, applying concepts from classical control engineering to quantum systems is not straightforward due to the peculiar role of measurement in quantum mechanics. This complication has been typically addressed by focusing on weak measurement of quantum systems captured via continuous stochastic differential equations \cite{carmichael2009open,geremia2003quantum,wiseman2009quantum}. In contrast, when considering projective measurement records, continuous stochastic filtering methods are no longer applicable as a quantum state is reset after each projective measurement.  Instead, many approaches for analyzing discrete measurements rely on rapid averaging or batch post-processing single-shot raw data \cite{proctor2020detecting,bravyi2020mitigating}.  The conversion of discretized measurement outcomes into continuous variables through these means ultimately discards useful time-domain information and adds a computational bottleneck that unnecessarily slows state-estimation algorithms.  Despite the importance of discrete measurement analysis in a wide range of applications for the characterization, calibration, and control of quantum systems,  single-shot projective measurements have not yet been directly incorporated into classical filtering techniques for quantum control. In particular, one requires a measurement model that is quantum mechanically accurate but also correctly captures the statistical properties of discrete observations of an otherwise continuous state-space. 

In this work we rigorously demonstrate how adaptive filtering incorporating quantum projective measurements can be understood through the theoretical framework of classical nonlinear filtering, and describe in full a set of algorithmic modifications enabling their use.  First, we demonstrate an efficient computational technique to discretize the amplitude domain of a classical signal in a manner that preserves statistical compatibility with classical filtering theory; this is achieved by combining Born's rule with an appropriate \textit{ansatz} for measurement noise in accordance with classical amplitude quantization. Second, we solve the resulting inference problem using a sequential Monte Carlo framework called particle filtering. Here, continuous probability distributions over a state space are approximately solved by a collection of discrete particles that undergo non-linear transformations. We introduce a novel set of rules to perform particle transformations in a manner compatible with single qubit projective measurements. 

With these modifications, we numerically validate that discretizing the continuous amplitudes of classical random signals, and performing discrete approximations to continuous probability distributions permits the convergence properties of classical nonlinear filtering carry over to classical filtering of discretized single-qubit projective measurements. Building on the experimental demonstration of adaptive spatial dephasing-field characterization first presented in Ref. \cite{gupta2020adaptive}, we use simulations to study the true error scaling coefficient with particle number. This error scaling coefficient is associated with the rate at which discrete empirical distributions tend to the true, continuous Bayesian posterior as the number of particles increase, thereby providing a numerical characterization of the convergence behaviour of the filter. We probe the properties of our proposed filter using two different types of numerical tests. First, our desired operation, where we run our novel particle filter using single-shot projective measurements. Here, the error scaling coefficient with particle number agrees with theoretically expected values from classical convergence analysis, which is unanticipated for the novel modifications introduced here for projective measurements. Second, we break our filter by uniformly randomizing (discarding) some state information before receiving the next single-shot projective measurement for every iteration of the filter. As the overall injection of random information progressively increases, the error scaling coefficient gradually increases above theoretically anticipated values. These empirical studies provide evidence that our methods approximate the true classical filtration generated by a sequence of projective measurements, and demonstrate model-robustness for a range of challenging operating conditions. 

\begin{figure}[t!]
 \centering
	\includegraphics[scale=.9]{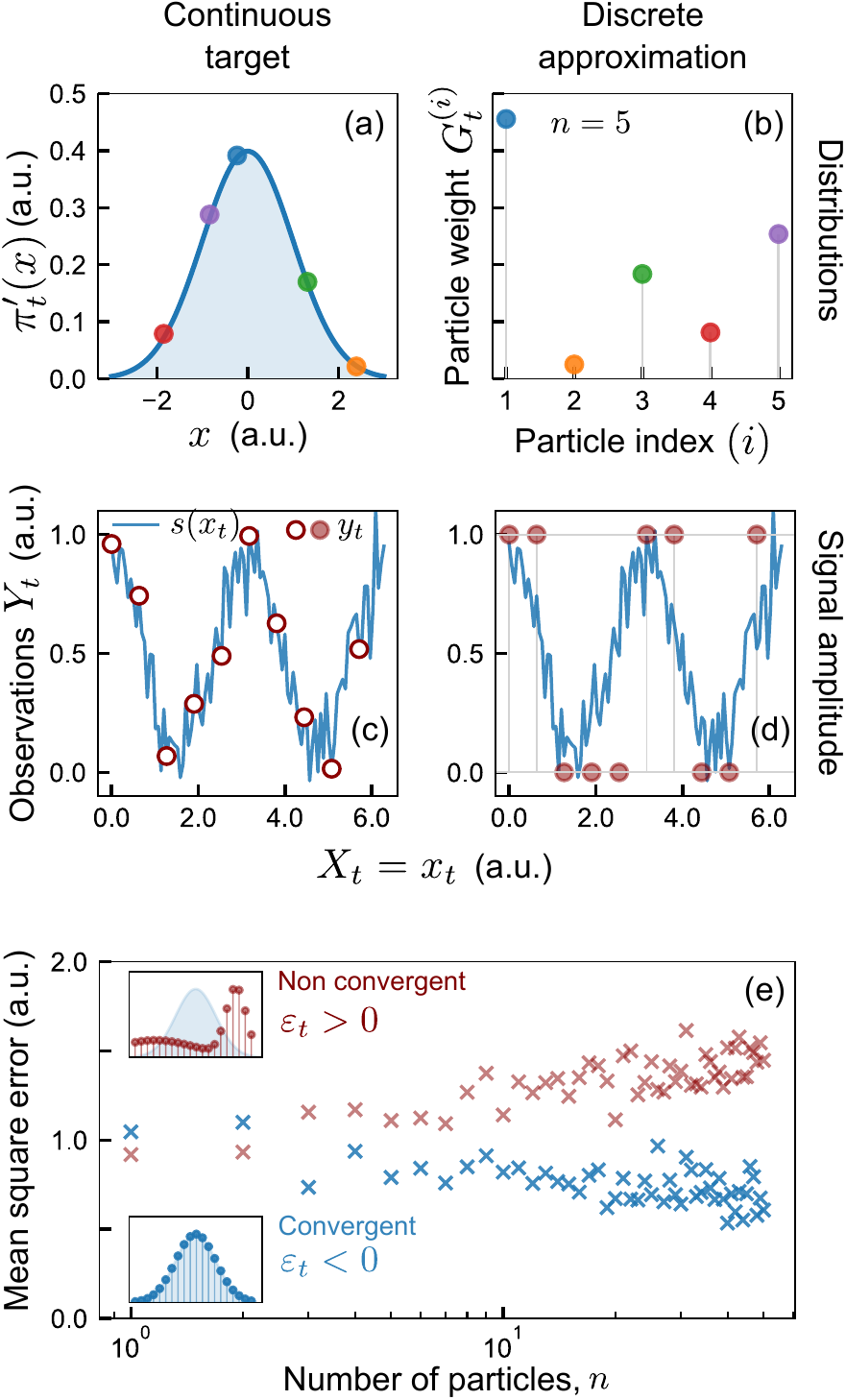}
	\caption[Schematic representation of continuous classical quantities and their discrete approximations]{\label{fig_schematic} Schematic overview of discrete approximations to continuous, classical quantities and their convergence behaviour. (a) Illustration of a true continuous probability density $\partial\pi_t/ \partial x $ with respect to a state space $x \in \mathbb{S}_X$ in 1D (blue solid line). Discrete particle positions (colored markers) for $n=5$ correspond to particle indices in (b).  (b) Illustrative plot of particle weights $G_t^{(i)}$ vs. particle index $i$ for some fixed $t$ computed from the particle filter.  Weights  are normalised, $\sum_{i=1}^n G_t^{(i)} = 1$, and the limit $n \to \infty$ recovers continuous distributions from empirical approximations in the mean-square limit. (c) Idealised depiction of discrete-time observations (red open circles) of a continuous signal (blue solid). (d) Discretization of signal amplitude in (c) into two levels yields binary measurements (red solid markers). Both discrete-time and discrete-amplitude properties of the signal are emphasized by vertical and horizontal grey lines respectively. (e) Schematic representation of a semilog plot of mean-square errors $\mathcal{L}_t$ vs. number of particles as $n$ increases for fixed but sufficiently large $t$. Convergent (non convergent) behaviour in blue (red) crosses corresponds to the extent of overlap of weighted particle positions and the true continuous distribution in the lower inset (upper inset). Error scaling coefficient $\varepsilon_t$  describes rate of change of errors as $n\to \infty$ for each $t$ and negative values for $\varepsilon_t$ indicate convergent behaviour.}
\end{figure} 

The manuscript is presented in the following parts. In \cref{sec:pf} we discuss the use of particle filters as a discrete approximation to continuous probability distributions in classical inference and their utility in solving difficult Bayesian inference problems. In \cref{sec:qsystems}, we outline how the quantum mechanical Born's rule can be combined with insights from classical discrete signal analysis so that individual projective measurement outcomes can be analyzed by classical filtering techniques. Subsequently, we  show that classical convergence properties of particle filters are retained even if discrete, projective measurements are used. In \cref{sec:nmqa}, we consider adaptive filtering with single-shot projective measurements, first presented in Ref.~\cite{gupta2020adaptive}, now discussed with greater generality and a focus on the convergence properties of filtering distributions in a general physical setting. Using the specific example of Ref.~\cite{gupta2020adaptive}, in \cref{sec:numericanalysis} we present numerical evidence for favourable convergence characteristics. Concluding remarks are provided in \cref{sec:Conclusion}. \FloatBarrier

\section{Particle filtering methods \label{sec:pf}}

Particle filters belong to a broader class of classical algorithms, known as sequential Monte Carlo algorithms, but have featured in quantum characterization and control applications. Quantum particle filters were developed in the context of continuous quantum measurements \cite{chase2009single}, while particle methods have also been used for adaptive Hamiltonian learning using projective measurements \cite{granade2012robust,granade2017structured}. Outside of quantum systems characterization, these methods have been popularized in nonlinear engineering control theory and probabilistic robotics, for example, in classical Simultaneous Localization and Mapping (SLAM) problems \cite{cadena2016past,durrant2006simultaneous,thrun2005probabilistic} where a robot must characterize (`map') and physically navigate through an unknown terrain. A common theme arising from these diverse applications is that particle filters perform strongly in high-dimensional, non-Gaussian and nonlinear state-spaces \cite{doucet2001introduction,candy2016bayesian,bergman1999recursive} that typically arise in context of characterising quantum systems.

The efficacy of these particle filtering methods in solving inference problems is due to their so-called particle branching mechanisms. These branching mechanisms are an essential part of assessing convergence, computational efficiency and correctness for a particle filter, irrespective of the specific details about measurement or system dynamics in any physical application. The subset of particle filters discussed here have extremely convenient convergence characteristics that can be exploited for designing algorithms for quantum control. In particular, a convenient convergence property is that the statistical behaviour of branching process determines the rate at which a particle filter converges to the true Bayesian posterior distribution as the number of particles increase \cite{bain2009}. Furthermore, these convergence characteristics do not place any major constraints on the dynamical evolution or measurement procedures for the system under consideration. This insight paves the way for using nonlinear classical filtering directly on discrete, single-shot outcomes obtained from quantum systems in a wide range of physical applications.

The key objective of any particle filter  is to approximate a true continuous Bayesian posterior distribution \cite{doucet2001introduction,candy2016bayesian,bergman1999recursive,murphy2000bayesian,poterjoy2016localized}. A true continuous Bayesian posterior distribution, denoted  $\pi_t$, is the conditional probability of observing $X_t$ given a set of measurements $Y_{0:t}$. The distribution $\pi_t$ is expressed as the conditional probability of $X_t$ given the \ofield{} generated by the observations $Y_{0:t}$.  In general, a transformation from $X_t \to Y_t$ is nonlinear, and in case of single-qubit measurements, the binary nature of $Y_t \in \{0,1\}$ further makes it difficult or impossible to derive an analytical filter update using measurement data.  

\begin{table}[h]
\centering
 \begin{tabular}{ | l | p{7.5cm} |} 
 \hline
 Sym. & Definition \\ \hline 
 $\mathbb{S}$ & A complete, separable metric (state) space for a R.V. \\
 $\mathcal{S}$ & The Borel $\sigma$-algebra \cite{szekeres2004course} generated by $\mathbb{S}$\\
 $C(\mathbb{S})$ & The space of real continuous functions on $\mathbb{S}$ \\
	$M(\mathbb{S})$ & The space of $\mathcal{S}$-measurable functions on $\mathbb{S}$ \\ 
	$B(\mathbb{S})$ & The space of bounded $\mathcal{S}$-measurable functions on $ \mathbb{S}$\\
	$C_b(\mathbb{S})$ & The space of bounded continuous functions on $ \mathbb{S}$\\
	$P(\mathbb{S})$ & The space of probability measures on $(\mathbb{S}, \mathcal{S})$ s.t. $\mu \in P(\mathbb{S})$ satisfies $\mu(\mathbb{S})=1$ \\
 \hline
 \end{tabular}
\caption{\label{table:notation} State, function and measure space notation (consistent with \cite{bain2009}). The abbreviation R.V. stands for any random variable and $\mathcal{S} \equiv \ofieldgen{\mathbb{S}}$ \emph{i.e.} the $\sigma$-field generated by the state space $\mathbb{S}$.}
\end{table} 
In the particle filtering approximation, the $i$-th particle represents a hypothesis about $X_t=x_t^{(i)}$, known as the `position' of the particle in $\mathbb{S}_X$, the state-space associated with $X_t$ (refer \cref{table:notation}). The collection of particle positions represents the empirical approximation to $\pi_t$. This approach permits a mechanism by which a filtering algorithm may be applied in order to obtain a numerical estimate of the posterior distribution by directly transforming particles at each iteration, rather than seeking analytical solutions using algebraic inversions or decomposition methods.  This discrete approximation, $ \pi_t^n$, for the true $\pi_t$, is expressed as
\begin{align}
 \pi_t^n &:= \frac{1}{n}\sum_{i=1}^n \delta_{x, x_t^{(i)}}, \label{eqn:25:main}
\end{align} where $n$ represent the total number of particles $\{x_t^{(i)}\}_{i=1}^n$, and each particle represents a hypothesis for $X_t$. In the above, the Kronecker delta, $\delta_{x, (\cdot)}$, is used because the approximate probability measures represent discrete probability distributions. This discrete approximation to a continuous distribution is schematically depicted in \cref{fig_schematic}(a), where a set of $n=5$ discrete particle positions are illustrated as points on a continuous probability density by colored circular markers. 

We now provide an overview of the particle-filtering algorithm.  During filtering, particles $\{x_t^{(i)}\}$ are transformed by operations which represent dynamical or measurement processes, represented by $K_t$ and the likelihood function $g_t$ respectively. In general the transformations represented by $K_t$ and $g_t$ are nonlinear, and the resulting transformed particles need not resemble the forms of analytic probability distributions. Under the additional assumption that $X$ is Markov, one uses the transition kernel for a Markov chain to obtain the distribution at $t$ if the distribution at $t-1$ is known. The result is called the predictive probability measure, $p_t$, or equivalently, the dynamical model for the filtering problem if drift characterization is relevant to a system under consideration. Thus, Bayes rule for the conditional probability of $X_t$ given observations $Y_{0:t}$ is written in the typical recursive form as,
\begin{align}
	\pi_t &:= g_t * p_t (A) := \frac{\int_A g_t(x) dp_t(x) }{p_t g_t }, \label{eqn:22:main}\\
	p_t &:= K_{t-1} \pi_{t-1}, p_t \in P(\mathbb{S}_X), \label{eqn:23:main}\\
	p_t g_t &:= \int_{\mathcal{S}_X} g_t(x) dp_t(x) > 0 \label{eqn:24:main}
\end{align} The use of the projective product in the first line, $*$, is essentially a restatement of Bayes rule. While the distribution $p_t$ is the true continuous predictive distribution, it can also be approximated by individually transforming particles in an empirical distribution, $ p_t^n := K_{t-1}\pi_{t-1}^n$. 

For each incoming measurement at $t$, a particle weight, denoted $G_t^{(i)}$, is computed for all $i=1,2, \hdots n$ particles. These weighted particles, $\bar{x}_t^{(i)}$, form the weighted distribution $ \bar{\pi}_t$, expressed as
\begin{align}
 \bar{\pi}_t^n &:= \sum_{i=1}^n G_t^{(i)} \delta_{x, \bar{x}_t^{(i)}}, \quad \bar{x}_t^{(i)} \sim p_t^n. \label{eqn:26m}
\end{align} Here, the particle weight, $G_t^{(i)}$ represents the probability of receiving a measurement $Y_t=y_t$ if the hypothesis captured by the $i$-th particle $X_t= \bar{x}_t^{(i)}$ is taken to be true. The bar notation $\bar{\cdot}$ indicates that the distribution should be computed after evolving particles from $t-1$ into the current iteration at $t$, and $G_t^{(i)}$ are calculated based on a single measurement, $Y_t$, received at $t$. As a schematic illustration, for each colored circular marker in \cref{fig_schematic}(a), the corresponding particle weight vs. particle index $i = 1, 2 \hdots n$ is shown in \cref{fig_schematic}(b).

The frequency with which a particle position is represented in the next generation at $t+1$ is proportional to its particle weight. Here, the weighted empirical distribution  of \cref{eqn:26m} is sampled according to the distribution of the particle weights $\{ G_t^{(i)} \}$. The resulting off-spring particles form posterior $\pi^n_t$; equivalently, the prior distribution for the iteration at $t+1$ and weights are reset to uniform. Over many iterations, particles with higher weights are represented more frequently in the particle populations enabling empirical particle distributions to gradually converge to the true, continuous Bayesian posterior. In words, as $n\to \infty$, the set of weights in \cref{fig_schematic}(b) approximate panel (a) for the schematic introduced earlier.  

The efficacy of any particle filter is linked to particle branching processes. These depend on how particle weights $\{ G_t^{(i)} \}$ are calculated whenever new measurement information is received, and how particles are propagated from one iteration $t$ to the next $t+1$ via re-sampling.  Branching mechanisms therefore form the core algorithmic representation of the approach to providing an approximate solution to the Bayesian inference problem. Many different particle branching processes have been proposed in particle filtering literature \cite{li2015resampling,beevers2007fixed,grisettiyz2005improving,godsill2001improvement}. We focus on branching processes satisfying \cref{prop:banchingproperties} which benefit from a number of well-understood convergence characteristics \cite{bain2009}. One of these characteristics is that the correctness of a particle filter can be analyzed as convergence to the true Bayesian posterior distribution as the number of particles $n$ increase for any $t$,
\begin{align}
    (\pi_t^n)_{n=1}^\infty \to \pi_t. \label{eqn:cstatement}
\end{align} This equation describes the convergence of a particle filter to the true Bayesian posterior. The arrow schematically depicted in \cref{eqn:cstatement} can be interpreted as `convergence in expectation' and `almost-sure convergence' of the empirical to the true Bayesian distributions. Both these forms of convergence describe the expected distance between two probability distributions, where the first focuses on distances between moments of a distribution and the second describes overall convergence between distributions, subject to several technical considerations \cite{bain2009}.  

A schematic illustration of different convergence behaviour is given in \cref{fig_schematic}(e). Here, the square error between the means of the target and estimated conditional distribution of  $X_t$ given data $Y_{0:t}$ increases (decreases) with particle number $n$ for a divergent (convergent) particle filter, as depicted in red (blue) crosses.  The consideration of higher order moments of the posterior distribution is omitted from this schematic figure. Upper and lower insets plot particle locations (circular markers) against a target continuous distribution (shaded  blue) for divergent vs. convergent particle filtering respectively, illustrating differences between the target and estimated distribution of  $X_t$ given data $Y_{0:t}$ during particle filtering.

We now explore the technical considerations under which convergent particle filtering can be guaranteed. To achieve convergence during particle filtering applications, the two important technical conditions are that $g_t$ is continuous and bounded and $K_t$ is Feller \cite{bain2009}. Once these conditions are satisfied, a particle filter possessing branching properties of \cref{prop:banchingproperties} is guaranteed to have both convergence in expectation and almost-sure convergence to the true Bayesian posterior as $n$ increases, as detailed in \cref{app:bkg:0}.

\begin{proposition} \label{prop:banchingproperties}
	Let $i$ denote a particle label with $i=1, 2, \hdots n$, $G_t^{(i)}$ denote a particle weight, and $\xi^{(i)}$ denote the frequency of a particle position at $t$. Branching mechanisms for a particle filter satisfy \cite{bain2009}:
	\begin{enumerate}
		\item Constant particle number $n = \sum_{i=1}^{n} \xi^{(i)}$ for all $t$.
		\item Conditional mean proportional to $G_t^{(i)}$, that is: $\ex{ \xi^{(i)} | \bar{\mathcal{G}}_t } = n G_t^{(i)}$.
		\item Conditional covariance matrix $(A^n_t)_{ij} := \ex{(\xi^{(i)} - n G_t^{(i)})^T (\xi^{(j)} - n G_t^{(j)}) | \bar{\mathcal{G}}_t }$ satisfy $q^T A^n_t q_t \leq nc_t$ for some constant $c_t$ and for any $n$ dimensional vector $q$ with entries $|q^{(i)} |< 1$.
	\end{enumerate}
\end{proposition} In \cref{prop:banchingproperties}, the quantity $\xi^{(i)}$ is the number of times the parent particle $x_t^{(i)}$ is copied and represented in the offspring generation of particles. The first proposition specifies that the total number of particles remains $n$ for all $t$ enabling a simpler analysis of the full branching random process from $t=0$ to $t$, \emph{i.e.} only the branching transitions within each $t$ need to be considered. The second proposition restates that empirical weight of the particle is associated with the true probability of observing that particle (state information) given some observed history via $\bar{\mathcal{G}}_t$. The last property places a constraint on the covariance matrix associated with the branching process. This constraint appears to have no \emph{a priori} justification, but it is a condition associated with a particle filter's convergence properties \cite{bain2009}.

In particular we focus on one aspect of convergence analysis which concerns the scaling behaviour of empirical distributions with particle number. This scaling behaviour can be associated with the behaviour of true errors generated during filtering irrespective of the system under consideration, and can be compared with the actual performance of particle filters in numerical simulations. Using \cref{prop:banchingproperties}, one can derive conditions on empirical filtering distributions as, 
\begin{align}
	 \ex{(p^n_tf - \pi_{t-1}^n (K_{t-1}f))^2} &\leq \frac{\parallel f \parallel_\infty^2}{n}, \label{eqn:31} \\
	\ex{(\pi_t^nf - \bar{\pi}_t^nf)^2} &\leq \frac{c_t \parallel f\parallel_\infty^2}{n}, \label{eqn:32} 
\end{align} for all $ f \in B(\mathbb{S}_X), t \geq 0$ (see \cref{table:notation}), where $\parallel f\parallel_\infty^2$ is a infinity norm for the function. These inequalities state that the expected distances between empirical distributions shrink as filtration proceeds.  The two different types of distances under consideration are firstly, from posterior distributions at $t-1$ to predictive distribution at $t$ (\cref{eqn:31}), and secondly, the empirical distribution before and after particle re-sampling within each $t$ (\cref{eqn:32}). The specific value of $c_t$  in \cref{eqn:31,eqn:32} depends on the type of branching process and its value cannot always be deduced \emph{a priori}. Supporting technical derivations for these equations and their relevance to the overall proofs for convergent particle filters is re-stated for completeness in \cref{app:bkg:0}. 

Of the branching processes satisfying \cref{prop:banchingproperties}, `bootstrap' filters are a popular example; an example algorithmic implementation is outlined in \cref{algorithm:bootstrap}. Here, particle weights are calculated only using the likelihood function $g_t(\cdot)$. In \cref{algorithm:bootstrap}, one sees that the empirical distributions of the bootstrap particle filter follow the progression,
\begin{align}
	K_{t-1}\pi_{t-1}^{n} \to \bar{\pi}_t^{n} \xrightarrow[n]{\text{re-sample}} \pi_t^{n} 	 \label{eqn:addin:1},
\end{align} and all empirical distributions have a constant particle number $n$. In the above, the first arrow represents computing particle weights using the likelihood function. The second arrow represents particle re-sampling as summarized by the last two lines of \cref{algorithm:bootstrap}. For bootstrap particle filters of \cref{algorithm:bootstrap}, the progression depicted in \cref{eqn:addin:1} is a multinomial branching process and it is theoretically tractable to show that $c_t=1$ (\cref{app:bkg:0}). 

\begin{algorithm}[H] 
	\caption{Bootstrap}\label{algorithm:bootstrap}
	\footnotesize
	\begin{algorithmic}[0] 
		\If{$ t = 0$}
		\State Sample $x_{0}^{(i)} \sim \pi_0, \quad i = 1, 2, \hdots n$ 
		\EndIf
		\If{$ t > 0$}
		\State Sample $\bar{x}_{t}^{(i)} \sim K_{t-1}\pi_{t-1}^n i = 1, 2, \hdots n$
		\State Receive $Y_t=y_t$; compute $G_t^{( i )} = \frac{g_t^{y_t}(\bar{x}_{t}^{(i)})}{\sum_{i}^{n} g_t^{y_t}(\bar{x}_{t}^{(i)})}$
		\State Replace $\bar{x}_{t}^{(i)}$ with $\xi^{(i)}$ offspring, such that $n= \sum_{i=1}^n\xi^{(i)}$
		\State Re-label offspring as $x_{t}^{(i)}$; reset $G_t^{( i )} = 1/n$ for $i = 1, \hdots, n$.
		\EndIf
	\end{algorithmic}	
\end{algorithm}

In the next section we introduce a model for projective measurements on quantum systems that may be employed in a data inference problem in which we must learn or estimate system dynamics based on a measurement record. Subsequently, we will proceed to establish a central result of this manuscript - that the use of the likelihood function introduced in \cref{sec:qsystems} for projective measurements on quantum systems does not disrupt the essential convergence properties of these particle filters.

\section{Nonlinear filtering of single-qubit measurements \label{sec:qsystems}}

In many physical settings, it is often desired that some continuous-valued classical process $X$ is inferred from a discrete-time sequence of measurements. The challenge posed by single-qubit measurements is that measurements can only assume certain allowed values, $Y \in \{0,1\}$. This challenge that a continuous $X$ can only be observed as discrete outcomes $Y$ is well-known in classical literature as the quantization of signal amplitude. The key insight described in this section is that the combination of an analytic prescription of Born's rule with classical amplitude quantization theory to describe single-shot projective measurement outcomes provides compatibility with any classical filtering algorithm. 

In classical signal processing, it is often the case that the continuous amplitude of some classical process is measured by a sensor that can only record discrete amplitude levels. In this context, amplitude-quantization theory specifies the statistical properties of discrete measurements of an otherwise continuous-amplitude signal \cite{widrow1996}. The amplitude of a classical signal is said to be discretized by $B$ bits if its continuous-amplitude is measured by a physical sensor which only has $2^B$ discrete amplitude-levels, up to some constant offset. These classical amplitude-quantized signals are analyzed via sampling a signal in the amplitude domain leading to an increase in the overall noise floor \cite{widrow1996}. We illustrate the procedure of amplitude quantization in \cref{fig_schematic}(c)-(d). An example of a continuous-amplitude discrete-time classical signal is first shown in panel (c) as red open markers, corresponding to discrete-time noisy measurements of a continuous-time signal (blue solid). This signal further undergoes a classical discretization of signal amplitude, where the $y$-axis is discretized  into two discrete levels ($B=1$ case). The resulting signal is a discrete-amplitude discrete-time signal given by the red filled markers in panel (d). 

For concreteness, we treat the case $B=1$ and extend this classical analogy to single-qubit projective measurements. We consider a classical signal consisting of a sequence of projective measurements. Let single-qubit states be expressed in the $\hat{\sigma}_z$ basis, and $\hat{U}(t', t; X)$ be some single-qubit unitary interaction that depends on $X$ for a qubit initially prepared in the ground state at the start of the procedure at $t'$. The Born probability for the outcome $Y_t$ of the projective measurement commenced at $t$ is $|\bra{Y_t} \hat{U}(t', t; X) \ket{0}|^2$ for $ Y_t \in \{0,1 \}$ and $t'< t$. Under these circumstances, the nonlinear measurement model for single-qubit measurements can be described as an outcome, $Y_t$, of a Bernoulli trial. This model is denoted in notation by the symbol $\mathcal{Q}(\cdot)$ for taking a biased coin flip with the bias given by the argument,
\begin{align}
Y_t &= \mathcal{Q} (|\bra{Y_t} \hat{U}(t', t; X) \ket{0}|^2), \label{eqn:2} \\ 
 \mathcal{Q}(z) &:= \mathrm{Binom}(p=z; n=1, k=1). \label{eqn:3}
\end{align} In the above, a binomial distribution has success probability $z$, number of trials $n=1$, and $k=1$ successes. These repeated single-shot measurements spaced $\Delta t$ apart gives rise to discrete classical random processes. We now interpret $t$ to be a discrete time index marking a set of repeated single-shot measurements, $Y_{0:t}:= \{Y_0, Y_1, \hdots, Y_t\}$ associated with $X_{0:t}:= \{X_0, X_1, \hdots, X_t\}$. The time step $\Delta t$ is set by total time for system preparation, interaction, measurement and reset, with $\Delta t$ typically much greater than the unitary interaction period in practical experiments. The slowly varying assumption on $X$ is that $\Delta t$ is much faster that any variation in $X$ and $X$ is approximately constant over the interaction $\hat{U}$.

The key observation is that some classical continuous amplitude $X_t$ yields only a discrete allowed value of $Y_t$ upon observation. If a sensor measures a continuous amplitude signal, $s(X_t)$, only as discrete allowed amplitude levels $Y_t \in \{0,1\}$, then this sensor has the overall effect of adding noise in the $t$-domain of the signal \cite{widrow1996,lipshitz1992quantization,karlsson2005,gustafsson2013generating}, represented by $v_t$. We express the association of the classical abstract signal with Born's rule as,
\begin{align}
	|\bra{Y_t} \hat{U}(t', t; X) \ket{0}|^2 \leftrightarrow s(X_t) + v_t + \frac{1}{2}, \label{eqn:4}
\end{align} where $v_t$ represents uncertainty in our knowledge of the true Born probability inferred from single shot measurements, and the term $\frac{1}{2}$ is an arbitrary global re-scaling factor so that $s(X_t) + v_t$ is zero mean for a single-shot measurement of a maximally mixed qubit state. For typical single-qubit measurements characterizing a quantum system, we assume commuting projective measurement procedures such that a joint probability density over the random variates $Y_t, s(X_t), v_t$ exist for all $t$. continuous

The statistical properties of $v_t$ determines how effectively one can incorporate discrete-amplitude measurements into conventional classical filtering by proposing an appropriate noise model for capturing uncertainty in single-shot measurement information. Examples of classical amplitude quantized sensor information assume a variety of models, for example, where signal distributions are convolved ($\star$) with a pulse train or uniform distribution, or integrated above and below each discrete amplitude-level and re-normalised to represent measurement errors  \cite{widrow1996,karlsson2005,gustafsson2013generating}. Our departure from these approaches is to consider that any uncertainty in our knowledge of the true Born probabilities arise from truncated error distributions representing amplitude quantization into two discrete levels. Specifically, we assume a noise model to be zero-mean Gaussian distribution $\mathcal{N}(0, \Sigma_v)$  with variance $ \Sigma_v$, which is convolved with a uniform $\mathcal{U}$ distribution as,
\begin{align}
 \measure{{}}{v_t} = \mathcal{N}(0, \Sigma_v) \star \mathcal{U}(a,b) , \quad \forall t, \quad \mathbb{S}_v = \mathbb{R}, \label{eqn:6}
\end{align} where $\measure{{}}{Z}$ represents a probability measure for the real-valued, random variate $Z$ defined over the space $z \in \mathbb{S}_Z$, and $a,b$ represents finite bounds on the values of these errors due to amplitude discretization. In the above, the notation $\measure{{}}{Z}$ is interpreted as a probability mass function over discrete values in $\mathbb{S}_Z$, or density over continuous values in $\mathbb{S}_Z$. 

Under these considerations, the continuous-amplitude measurement model takes the form
\begin{align}
 &\measure{{}}{Y_t=y | s(X_t)} \nonumber \\
 &= \int_{\mathbb{S}_v} \left( (s(X_t) + v_t)(\delta_{y,1} - \delta_{y,0}) + \frac{1}{2} \right) \measure{{}}{v_t}(v) dv, \label{eqn:5}
\end{align} where $\delta_{x,y}$ takes the value $1$ if $x=y$ or zero otherwise. As one example of the noiseless ideal case $v_t\equiv 0$, the ideal Born probability of observing the qubit in $\ket{1}$ is then $s(X_t) + \frac{1}{2}$ or $\frac{1}{2} - s(X_t)$ for observing the qubit in the $\ket{0}$ state.

Substituting \cref{eqn:6} into \cref{eqn:5}, and performing the relevant integration yields the final form of likelihood function under amplitude-discretization of Born probabilities, 
\begin{align}
 \measure{{}}{Y_t=y | s(X_t)}  &= \frac{\rhoq{0}}{2} + \rhoq{0} s(X_t) \left( \delta_{y,1} - \delta_{y,0}\right), \label{eqn:7} \\
  & \equiv g_t^{Y_t=y_t}(X_t), \nonumber 
\end{align} with the real-valued scalar $\rho_0$ obtained from integration as,
\begin{align}
 \rhoq{0}	= \erf(\frac{2b}{\sqrt{2\Sigma_v}}) + \frac{\sqrt{2\Sigma_v}}{2b} \frac{e^{-(\frac{2b}{\sqrt{2\Sigma_v}})^2}}{\sqrt{\pi}} - \frac{1}{2b}\frac{\sqrt{2\Sigma_v}}{\sqrt{\pi}} \label{eqn:8}, 
\end{align} where $\erf$ is the error function with values between $[-1, 1]$. For error sources that are symmetric with respect to how they affect single-qubit states, one sets $-a=b$ in the calculation above. Asymmetric error distributions $-a \neq b$ may arise, for example, when noise during state-detection depends on the state of qubit at the start of a projective measurement procedure \emph{e.g.} state-dependent decay of hyperfine qubits in trapped-ion quantum computers \cite{olmschenk2007manipulation,ejtemaee2010optimization}, but are not treated in this manuscript. In circumstances when $b \leq 3 \Sigma_v$, model failure may occur as information in the original distribution is being discarded by the procedure for amplitude discretization. 

While we have focused on single-qubit measurements with two possible discrete amplitude levels (`0' or `1'), one may extend to $B$-qubit measurements with $2^B$ discrete levels if these states are individually discernable in experiments. In all of these cases, we assume that $s(X_t)$ is continuous such that the properties of the resulting discrete-amplitude signal can be described via methods of Refs.~\cite{widrow1996,karlsson2005}. Additionally, we will also assume that $s(\cdot)$ is bounded and $s(\cdot)$  has an inverse $s^{-1}$() on $\mathbb{S}_{X_t}$, the space of allowed continuous values for $X_t \in \mathbb{S}_{X_t}$. As discussed below, the boundedness property ensures that our likelihood function can be safely incorporated into bootstrap particle filtering while preserving convergence properties of these filters. Subsequently, in \cref{sec:nmqa}, the inverse $s^{-1}$ is used to share estimated state information in small regions for adaptive particle filtering.

We now establish that the likelihood for projective measurements proposed here can be incorporated within bootstrap particle filtering without affecting standard convergence theorems. Our likelihood function is given by \cref{eqn:7,eqn:8}. As discussed in \cref{sec:pf}, this function needs to be bounded and continuous for conventional convergence properties of bootstrap particle filters to hold. Examinining \cref{eqn:8}, we see that for $b \neq 0$, the scalar value $\rhoq{0}$ is bounded, as evident by considering the following two limiting cases. The limit $\Sigma_v \to 0$, the scalar $\rhoq{0} \to 1$ and the ideal case of a coin flip with the win probability give by Born's rule is obtained. In the opposite limit, $\Sigma_v \to \infty$, the scalar $\rhoq{0} \to 0$ and no inference is possible. Assuming $s(\cdot)$ is bounded and $b\neq 0$, the proposed likelihood function $g_t^{Y_t=y_t}(X_t)$ is also bounded. For continuity, it is required that the likelihood function is continuous over the state space of $X_t$ for a specific instance of data $Y_t$ \cite{crisan2002survey,karlsson2005,bain2009}. For a fixed instance, $Y_t=y_t$, $g_t^{Y_t=y_t}(X_t)$ in \cref{eqn:7} is expressed by either $\rhoq{0}/2 + \rhoq{0} s(X_t)$ or $\rhoq{0}/2 - \rhoq{0} s(X_t)$. Assuming $s(\cdot)$ is continuous, the proposed likelihood function $g_t^{Y_t=y_t}(X_t)$ is also continuous with respect to $X_t$ for an instance of $y_t$. Thus for continuous and bounded $s(\cdot)$ and $b\neq 0$, the proposed likelihood satisfies the key properties required for conventional convergence properties. Further our result is general in the sense that aside from the observation process, no further information about the physical application, system dynamics or the noise environment is being assumed.  The specific case $b=1/2$ is considered in the remaining sections.

So far our work allows quantum projective measurements to be analyzed by fully exploiting the power of particle techniques for non-linear, non-Gaussian, high-dimensional state-spaces typically arising in the context of quantum characterization problems. Next, we provide technical details about an adaptive filtering framework as a variant of bootstrap particle filter with multinomial branching. However, our framework departs substantially from traditional bootstrap filters as it incorporates features for adaptive control using quantum projective measurements. We outline these features in the next section and subsequently investigate the numerical error scaling behaviour of our protocol with the $1/n$ behavior predicted by \cref{eqn:31,eqn:32}.

\section{ Adaptive filtering for quantum systems characterization \label{sec:nmqa}} 

Our challenge in this section is to outline a theoretical framework capable of adaptively characterizing and predicting classical correlations arising in projective measurement records. These classical correlations may arise, for instance, due to the interaction of the quantum system with its ambient environment, unanticipated system-dynamics, or intrinsic performance variations or noise in hardware. However, a naive application of multivariate filtering techniques to projective measurement records, even of commuting quantum observables, presents several difficulties. One issue is that quantum projective measurements are inherently local. In particular, Born's provides an unambiguous link between the measurement information and the elements of some multivariate $X$ being inferred. In the language of classical estimation and mapping, this statement means that there is often no immediate benefit in defining a joint, classical Bayesian inference problem over elements of a multi-variate $X$ in filtering, a stark contrast to related classical literature, for example, for simultaneous localization and mapping (SLAM) applications \cite{durrant2006simultaneous,bailey2006simultaneous,thrun2005probabilistic}. Below, we present a deeper analysis of the implications of our adaptive methods, as first presented in \cite{gupta2020adaptive}, on overcoming these challenges and on filter convergence.

To accommodate predictive estimation of classical correlations in projective measurement records, we now associate points in some classical (continuous) parameter space with an index, $j$, as well as the discrete sequencing index $t$.  This parameter space may arise in different physical applications where classical, continuous variables are sparsely sampled, as examples, due to geometric  arrangement of qubits in space \cite{gupta2020integration}; the choice different measurement procedures, tomography of continuous-variable systems \cite{landon2018quantitative} or noise spectroscopy. For a $d$-dimensional observation vector, we assume a classical joint probability distribution must exist over all $d$ elements of $Y_t^{(j)}, j=1, 2, \hdots d$, \emph{i.e.} quantum mechanical observables associated with $Y_t^{(j)}$ commute for all $j$ and $t$. If the labels $j=1, 2, \hdots, d$ are measurements of \textit{different} points in this parameter space, then each observation $Y_t^{(j)}$ is local and provides information only about the elements of $X_t^{(i)}$ uniquely associated with the label $j$ at iteration $t$. If instead the labels $j=1, 2, \hdots, d$ are repeated measurements of the \textit{same} point in parameter space, then the empirical mean of repeated measurements $\frac{1}{d}\sum_{j=1}^{d} Y_t^{(j)}$ is the empirical Born probability. 

In order to efficiently learn classical correlations in projective measurement records and overcome these technical challenges, the adaptive filtering framework of this section shares estimated state-information between elements of $X$ during filtering, while behaving in accordance with the branching properties of \cref{prop:banchingproperties}. Our key observation is that many physical settings and noise sources lead to classical, continuously-varying phenomena in $j$ and $t$. In our framework the outputs of classical state estimation at one coordinate point associated $j$ can be spread locally about that location. The region or neighborhood within which information-sharing occurs can also be estimated as part of the particle filtering process. Thus in the language of classical mapping problems, for each $t$ we estimate both map values at the point $j$ and approximate map gradients in small regions about $j$. The resulting output of the particle filter is a characterization of classical correlations over parameter space indexed by $j$ and $t$ using projective measurement records.

Framed in the language of classical map-building, a true state vector, $X_t$, contains both the register of map values $F_t$ and local approximate map gradient information, $R_t$, that is, 
\begin{align}
X_t &= \begin{bmatrix} F_t^{(1)} & \hdots F_t^{(d)} & R_t^{(1)} & \hdots R_t^{(d)} \end{bmatrix} = \begin{bmatrix} F_t & R_t\end{bmatrix}, \label{eqn:33}\\
\mathbb{S}_{X} & = \mathbb{S}_{F} \times \mathbb{S}_{R}, \forall t, \label{eqn:34} \\
\mathbb{S}_{F} & = [F_{\min}, F_{\max}], \quad \mathbb{S}_{R} = [R_{\min}, R_{\max}] \label{eqn:35}.
\end{align} In the above the quantities $F_t$ and $R_t$ represent $d$-dimensional, real, continuous vector-valued random variables, and their outcomes take values between $F_t^{(j)} \in [F_{\min}, F_{\max}]$ and $R_t^{(j)} \in [R_{\min}, R_{\max}]$ for any location in parameter space labelled $j= 1, 2, \hdots d$. 

Unlike typical particle filtering, our algorithm locally estimates the value of the field for a measured point at $j$, before sharing this information with neighboring points in the vicinity of $j$. The algorithm is responsible for determining the appropriate size of circular neighborhoods of radius $R_t^{(j)}$ about the point labeled by $j$. The set of points inside the neighborhood, $Q_t$ shrinks or grows about $j$ as the autonomous inference process progresses. Under these circumstances, this adaptive filtering protocol incorporates not only a local physical single-qubit projective measurement at $j$ using \cref{eqn:4}, 
\begin{align}
Y_t^{(j)} = &\mathcal{Q} (s(F_t^{(j)}) + v_t + \frac{1}{2} ), \quad \mathbb{S}_Y = \{0, 1\}, \label{eqn:39}
\end{align} but also data-messages generated by $j$ for locations $q_t$,
\begin{align}
\hat{Y}_t^{(q_t)} =& \mathcal{Q} (s(\chi_t^{(j, q_t)}) + \frac{1}{2}),\quad \mathbb{S}_{\hat{Y}} = \{0, 1\}, \label{eqn:40} \\
&\forall q_t \in Q_t^{(j)}. \nonumber
\end{align} In the above, $\chi_t^{(j,q_t)}$ is a convex combination of the existing estimate at $q_t$  and new information due to a measurement $Y_t^{(j)}$ received at $j$. The calculations associated with the term  $\chi_t^{(j,q_t)}$ invoke continuity of physical phenomena whereby new information at $j$ is shared over a region about $j$ via any choice of a sigmoidal function \cite{ito1992approximation}, here set to be a Gaussian function, parameterized by the estimate of $R_t^{(j)}$. The term $\chi_t^{(j,q_t)}$ is computed using the posterior information at $t$ and has the effect of introducing correlations between the elements of particle positions in the next iteration $t+1$. Detailed technical information is provided in \cref{app:bkg:3} for completeness. 

Having modified conventional filtering with this information-sharing mechanism, we now focus on the branching properties of this framework and any potential implications on convergence properties of typical particle filters. In particular, two different types of particle species are used by the filter within a bootstrap filtering structure. Let $\alpha$-particles be a set of $n_\alpha$ number of particles. For each parent $\alpha$-particle, let $\beta_\alpha$-particles be a set of $n_\beta$ number of daughter particles useful for enabling neighbourhood discovery and adaptation during filtering. The layer of $\alpha$-particles, $ \{ x_t^{(\alpha)}\}_{\alpha=1}^{n_\alpha} $ carry a hypothesis about $X_t$, 
\begin{align}
	 x_t^{(\alpha)} = \begin{bmatrix} f_t^{(\alpha)} & r_t^{(\alpha)}\end{bmatrix}, \label{eqn:41}
\end{align} where lowercase $x_t, f_t, r_t$ refer to instances of the true process in uppercase $X_t, F_t, R_t$. Additionally in \cref{eqn:60}, $\beta_\alpha$-particles are a set of $n_\beta$ number of particles for each of the $n_\alpha$ parents. A single $\beta_\alpha$-particle carries a hypotheses for $R_t^{(j)}$ assuming that $F_t$ and neighborhoods at other locations $R_t^{(j'\neq j)}$ are known, expressed in our notation as,
\begin{align}
	x_t^{(j, \alpha, \beta_\alpha)} = \begin{bmatrix} r_t^{(j, \alpha, \beta_\alpha)}\end{bmatrix}, \label{eqn:42}
\end{align} where the distribution of $\beta$-particles is the conditional distribution of $R_t^{(j)}$ given $X_t \setminus R_t^{(j)}$. Here, the superscript notation ${}^{(j, \alpha, \beta_\alpha)}$ refers to the location label $j$ for the parent $\alpha$-particle index, $\alpha$, and its associated $\beta_\alpha$-particle. This empirical distribution of $\beta_\alpha$-particles is related to the parent $\alpha$-particle using the empirical mean,
\begin{align}
	r_t^{(j, \alpha)} = \mathbb{E}_{\beta_\alpha}\left[{ x_t^{(j, \alpha, \beta_\alpha)}}\right]. \label{eqn:43}
\end{align} The expression above relates the empirical mean of the $\beta_\alpha$-particles for each parent $\alpha$-particle to the element $r_t^{(j, \alpha)}$. 

These manipulations lead to the following progression of empirical distributions for each $t$. 
\begin{align}
	K_{t-1}\pi_{t-1}^{n_\alpha} \to \bar{\pi}_t^{(j, n_\alpha n_\beta )} \xrightarrow[N_1=n_\alpha]{\text{re-sample}} \bar{\pi}_t^{(j, n_\alpha)} \xrightarrow[N_2=n_\alpha]{\text{re-sample}} \pi_t^{n_\alpha}, 	 \label{eqn:59}
\end{align} where the index $j$ makes explicit that each iteration $t$ receives physical measurements at the label $j$, and the superscript $n_\alpha n_\beta$ (or $n_\alpha$) indicates the total number of particles in the weighted distribution, $\bar{\pi}_t^{(j, \cdot)}$.   Two re-sampling steps are required to move from $\bar{\pi}_t^{(j, n_\alpha n_\beta )} \to \bar{\pi}_t^{(j, n_\alpha)} \to \pi_t^{n_\alpha}$ corresponding to the arrows, where $N_1$ and $N_2$ represent the total number of particles in the new generation after re-sampling.  This progression of empirical distributions of \cref{eqn:60,eqn:61,eqn:62} in \cref{eqn:59} can be compared  to the bootstrap particle filter in \cref{eqn:addin:1}, where these measures are expressed as
\begin{align}
	\bar{\pi}_t^{(j, n_\alpha n_\beta )} &:= \sum_{\alpha=1}^{n_\alpha} \sum_{\beta_\alpha=1}^{n_\beta} G_t^{(j,\alpha, \beta_\alpha)} \delta_{x,\bar{x}_t^{(j, \alpha, \beta_\alpha)}} \label{eqn:60}, \\
	\bar{\pi}_t^{(j, n_\alpha)} &= \sum_{\alpha=1}^{n_\alpha} \Omega_t^{(j, \alpha)} \delta_{x,\bar{x}_t^{(\alpha)}}, \label{eqn:61}\\
	\pi_t^{n_\alpha} &= \sum_{\alpha=1}^{n_\alpha} \delta_{x, x_t^{(\alpha)}} \label{eqn:62}.
\end{align} In \cref{eqn:60}, the particle weights $G_t^{(j, \alpha, \beta_\alpha)}$ are computed using a scoring function $g_t^{(y_{t}^{(j)}, \alpha, \beta_\alpha)}(\lambda_1, \lambda_2, \Lambda_t)$,
\begin{align}
	G_t^{(j, \alpha, \beta_\alpha)} &:= g_t^{(y_{t}^{(j)}, \alpha, \beta_\alpha)}(\lambda_1, \lambda_2, \Lambda_t), \label{eqn:63}
\end{align} which incorporates the likelihood function of \cref{sec:qsystems} and whose form and parameters are introduced in full in the Appendices as \cref{eqn:46}. The weights $G_t^{(j, \alpha, \beta_\alpha)}$ are rearranged into new weights $ \Omega_t^{(j, \alpha)}$ after the first re-sampling step in \cref{eqn:61}. The use of the bar notation, $\bar{{}}$, indicates that posterior particle positions at $t-1$ have been propagated by the transition kernel to the step $t$, as indicated by sequence in \cref{eqn:59}.

Using the empirical definitions above, the pseudo-code summarising our proposed framework is given in \cref{algorithm:nmqa}. As with standard particle filters, our algorithm is initiated by sampling from a prior distribution. At any iteration $t$, all particles from the posterior distribution at $t-1$ are propagated to $t$ via the transition kernel $K_t$ in (i). Upon receiving measurements and data messages in step (ii), particles are subsequently scored using the likelihood function in (iii)-(iv) in a manner similar to bootstrap particle filtering. The subsequent steps involve particle re-sampling steps and adaptive control actions. In particular, steps (vi) corresponds to computing an the empirical variance estimate with respect to the $\beta_\alpha$-particles for each $\alpha$. The resulting quantity, $C_t^{(k)}$ for $k=1, 2, \hdots d$, is a Fano factor and it is used in the control step (xi) by scheduling the $t+1$ physical measurement for the label $j'$ associated with maximal uncertainty $j' = \mathrm{argmax}_k \{C_t^{(k)}\}_{k=1}^d$. Individual calculation steps for our code are fully specified in \cref{app:bkg:3}.

\begin{algorithm}[H] 
	\caption{Adaptive filtering for quantum msmts.}\label{algorithm:nmqa}
	\footnotesize
	\begin{algorithmic}[0] 
		\If{$ t = 0$}
		\State Sample $x_{0}^{(i)} \sim \pi_0, \quad i = 1, 2, \hdots n$ 
		\EndIf
		\If{$ t > 0$}
		\State (i) Sample $\bar{x}_{t}^{(i)} \sim K_{t-1}\pi_{t-1}^n, i = 1, 2, \hdots n_\alpha$
		\State (ii) Receive $Y_t^{(j)}=y_t^{(j)}$; generate $\beta_\alpha$-particles from \cref{eqn:69} or \cref{eqn:70}
		\State (iii)-(v) Compute $G_t^{(j, \alpha, \beta_\alpha)}$ using $g_t^{(y_{t}^{(j)}, \alpha, \beta_\alpha)}(\lambda_1, \lambda_2, \Lambda_t)$
		\State (vi) Replace $\bar{x}_{t}^{(\alpha, \beta_\alpha)}$ with $\xi_t^{(\beta')}$ offspring; $N_1 = \sum_{\beta'=1}^{n_\alpha n_\beta} \xi_t^{(\beta')}$. Reset to uniform weights $1/N_1$.
		\State (vii) Store $r_t^{(j, \alpha)}$ and $C_t^{(j)}$ from surviving particle-pairs
		\State (viii) Compute $\Omega_t^{(j, \alpha)} = \frac{\text{ num. of $\beta_\alpha$ survivors}}{N_1}$; discard $\beta_\alpha$-particles, $\forall \alpha$ 
		\State (ix) Replace $\bar{x}_{t}^{(\alpha)}$ with $\eta_t^{(\alpha)}$ offspring; $ N_2 = \sum_{\alpha=1}^{n_\alpha} \eta_t^{(\alpha)}$. Reset to uniform weights $1/N_2$.
		\State (x) Re-label surviving particles as $x_{t}^{(\alpha)}$, for $\alpha = 1, 2, \hdots n_\alpha$
		\State (xi) Schedule next measurement $j = \mathrm{argmax}_k C_t^{(k)}$
		\State (xii) Generate \& update $\hat{Y}_t^{(q_t)}$, for all $q_t \in Q_t^{(j)}$
		\EndIf
	\end{algorithmic}	
\end{algorithm}

Under this adaptive protocol, we seek the convenient convergence properties of particle filtering discussed in earlier sections and we  discuss the extent to which particle branching in \cref{algorithm:nmqa}, satisfies of \cref{prop:banchingproperties}. The following progression of empirical distributions in \cref{algorithm:nmqa}, 
\begin{align}
 \bar{\pi}_t^{(j, n_\alpha n_\beta )} \to \bar{\pi}_t^{(j, n_\alpha)} \to \pi_t^{n_\alpha},
\end{align} is found to be a multinomial process similar to conventional particle filtering. This process represents a rearrangement of particle weights into the weights $ \Omega_t^{(j, \alpha)}$ and it forms a multinomial random process if particle number is conserved during each re-sampling step. To see this, let $\beta'$ be the labels over all particle-pairs $(\alpha, \beta_\alpha)$ so that these labels correspond to the indices $\beta' = 1,2,\hdots n_\alpha n_\beta$. Let $A_\alpha$ be the grouping of $\beta_\alpha$-particle weights for each $\alpha$-parent, where $\alpha$ is the label over parent particles $ 1,2,\hdots n_\alpha$ as before. This means that the labels $\beta'$ are partitioned into $n_\alpha$ non-overlapping categories. Then the weights for each $A_\alpha$ category are
\begin{align}
	\Omega_t^{(j, \alpha)} := \sum_{\beta' \in A_\alpha} G_t^{(j, \beta')}. \label{eqn:64}
\end{align} The re-categorization given by the equation above occurs in (viii) of \cref{algorithm:nmqa} where the weights $\Omega_t^{(j, \alpha)}$ are proportional to the count over surviving $\beta_\alpha$ particles for each parent. Additionally, for $N_1=N_2=N$, the number of offspring in each re-sampled offspring generation satisfy
\begin{align}
	N &= \sum_{\beta'=1}^{n_\alpha n_\beta} \xi_t^{(\beta')}, \label{eqn:65}\\
	&= \sum_{\alpha=1}^{n_\alpha} \sum_{\beta' \in A_\alpha} \xi_t^{(\beta')},\label{eqn:66} \\
	&= \sum_{\alpha=1}^{n_\alpha} \eta_t^{(\alpha)}, \quad \eta_t^{(\alpha)} = \sum_{\beta' \in A_\alpha} \xi_t^{(\beta')}.\label{eqn:67}
\end{align} Thus, the re-sampling steps represent a re-categorization of $\beta_\alpha$ particle weights into non-overlapping sets associated with each $\alpha$-parent. If particle number is conserved $N_1=N_2=n_\alpha$, then two consecutive particle re-sampling steps in \cref{algorithm:nmqa} are multinomial and satisfy \cref{prop:banchingproperties}. These observations establish the second result that the adaptive filter of \cref{algorithm:nmqa} shares the same multinomial particle re-sampling process consistent with traditional particle filtering in \cref{algorithm:bootstrap}. 

The departure of \cref{algorithm:nmqa} from conventional particle filtering lies in the following step of \cref{eqn:59},
\begin{align}
	K_{t-1}\pi_{t-1}^{n_\alpha} \to \bar{\pi}_t^{(j, n_\alpha n_\beta )}, \label{eqn:68}
\end{align} which requires the generation of $\beta_\alpha$-particles by creating samples of $r_t^{(j, \alpha, \beta_\alpha)}$ at the start of each iteration $t$. We propose two methods for particle generation - `Uniform' or `Trunc. Gauss'. 
The Uniform method resets all $\beta_\alpha$ particles to the initial distribution for $R_0$ at any $t$ or parent index $\alpha$,
\begin{align}
	r_t^{(j, \alpha, \beta_\alpha)} &\sim \mathcal{U}(\mathbb{S}_R), && \forall t, r_t^{(j, \alpha, \beta_\alpha)} \in \mathbb{S}_R. \label{eqn:69}
\end{align} This method represents a strong breakdown of the transfer of estimated state information about $R_t$ from $t$ to $t+1$ during the estimation procedure. In contrast, Trunc. Gauss preserves some information about the estimated $R_t$ from $t$ to $t+1$ for each parent index $\alpha$,
\begin{align}
	r_t^{(j, \alpha, \beta_\alpha)} &\sim \mathcal{N}(\bar{r}_{t}^{(j, \alpha)} , \bar{r}_{t}^{(j, \alpha)} C_{t-1}^{(j)} ), && \forall t, r_t^{(j, \alpha, \beta_\alpha)} \in \mathbb{S}_R. \label{eqn:70}
\end{align} Here, one uses the approximation that the true distribution of $R_t$ at each step can be summarized by the first two moments of a truncated Gaussian distribution. Secondly, one assumes that $ \bar{r}_{t}^{(j, \alpha)} C_{t-1}^{(j)}$ is an appropriate approximation for the true second moment of $R_t$. As before, the barred quantities $\bar{r}_{t}$ denote that the posterior information at $t-1$ have been propagated into the current $t$ via the transition kernel $K_t$. 

The impact of these departures on the convergence properties for particle filtering are now investigated numerically in the next section. In particular, the scaling behavior of true errors with particle number $n$ will be explored via simulating the specific example in Ref. \cite{gupta2020adaptive}.

\section{\label{sec:numericanalysis}Numerical analysis}

In previous sections, we discussed the convergence properties of particle filters. For the specific case of bootstrap particle filters with multinomial re-sampling, of which both \cref{algorithm:bootstrap,algorithm:nmqa} are examples, the condition $c_t=1$ means that expected value of distance between the two empirical probability measures in \cref{eqn:31,eqn:32} decays as $\frac{1}{n}$ as $n\to \infty$ almost surely, where $n$ is the particle number. However, \cref{algorithm:nmqa} additionally accommodates both single-qubit measurements (\cref{sec:qsystems}) and adaptive control features (\cref{sec:nmqa}) that depart substantially from conventional filtering literature and thus it is not at all clear if the scaling behaviour predicted by conventional convergence theory apply here. In this section, we numerically analyze whether the true error scaling behaviour of \cref{algorithm:nmqa} with particle number accords with predictions from conventional convergence theory. 

Instead of comparing the distance between empirical measures for the state $X_t$ in \cref{eqn:31,eqn:32}, in our analyses, we focus on the first moments associated with these empirical measures and compare the posterior estimate of $F_t$ from the algorithm with the true $F_t$ using simulations. Under these approximations, let $\mathcal{L}_t $ be the expected value of the true mean-square error per label $j$ at iteration $t$. From \cref{eqn:32}, let $\mathcal{L}_t $ scale with particle number $ n \equiv n_\alpha$ according to the following postulated relationship,
\begin{align}
 \log \mathcal{L}_t = \varepsilon_t \log n_\alpha,
\end{align} where $\varepsilon_t$ is a real-valued scalar for finite values of $t$. Broadly, a value of $\varepsilon_t < 0$ indicates that expanding the particle number improves the inference procedure (error decreases with greater ${n_\alpha}$), while a value $\varepsilon_t > 0$ indicates increased error with ${n_\alpha}$. We expect $\varepsilon_t \in [-1, 0)$ for an algorithm that accords with conventional convergence theory, where $\varepsilon_t = -c_t = -1$ holds if multinomial re-sampling satisfying \cref{prop:banchingproperties} is the slowest contribution to overall algorithmic convergence. 

For this empirical analysis, we focus on a specific example presented in  Refs.~\cite{gupta2020adaptive,gupta2020integration}). In this example, one assumes that a set of independent qubits are subject to a classical, externally applied dephasing field, $k \in \{1, 2, \hdots d\}$ labels coordinate positions of qubits in 2D space, and $s(\cdot)$ is given by a relative phase single-qubit Ramsey measurement. For classical dephasing in Ramsey measurements, 
 \begin{align}
 	s(F_t) \equiv \frac{1}{2} \cos(F_t), \quad \mathbb{S}_F:= [0, \pi], \label{eqn:71}
 \end{align} where $F_t$ has the physical interpretation of qubit phases at each location, giving rise interference between quantum amplitudes of single-qubit states. In Ramsey measurements, the form of $s(\cdot)$ is nonlinear, bounded, continuous over the half cycle $[0, \pi]$, and this gives rise to a nonlinear particle filtering problem discussed in previous sections. Substituting \cref{eqn:71} into \cref{eqn:4} gives the measurement model for single qubits under dephasing,
 \begin{align}
 	Y_t &= \mathcal{Q} ( \frac{1}{2} \cos(F_t) + v_t + \frac{1}{2} ), \quad \mathbb{S}_Y = \{0, 1\}. \label{eqn:10}
 \end{align} We further assume that a slowly-varying $X_t$ is sampled rapidly using measurement data, and the approximation \begin{align}
 K_{t}(x, A) = \measure{{}}{X_{t+1} \in A | X_{t}=x} = \delta(x), \label{eqn:18}
\end{align} is used, where the symbol $\delta(x)$ is interpreted as the Dirac-delta at $x$. The equation above defines the assumption that $F_t$ is approximately static relative to a high measurement sample rate in $t$. 

By taking the first $d$ elements of the posterior $X_t$, and comparing it with a true dephasing field used during simulations, the error is computed as
\begin{align}
	\mathcal{L}_t := \ex{ ||( \mathbb{E}_\alpha[f_t^{(\alpha)}] - F^*)||_2^2 / d}. \label{eqn:75} 
\end{align} In the above, the true map used in simulations is the vector-valued $F^*$, $\ex{\cdot}$ is an expectation taken over repetitions of each simulation, $\mathbb{E}_\alpha[\cdot]$ captures the first moment of the particle distribution. Specifically $\mathbb{E}_\alpha [f_t^{(\alpha)}]$ are the first $d$ elements of the posterior estimated $X_t$ obtained as the mean of the posterior particle positions at $t$, and $||\cdot ||_2^2$ represents the total squared error on all $d$ locations. 

\begin{figure}[t!]
 \centering
	\includegraphics[scale=1.0]{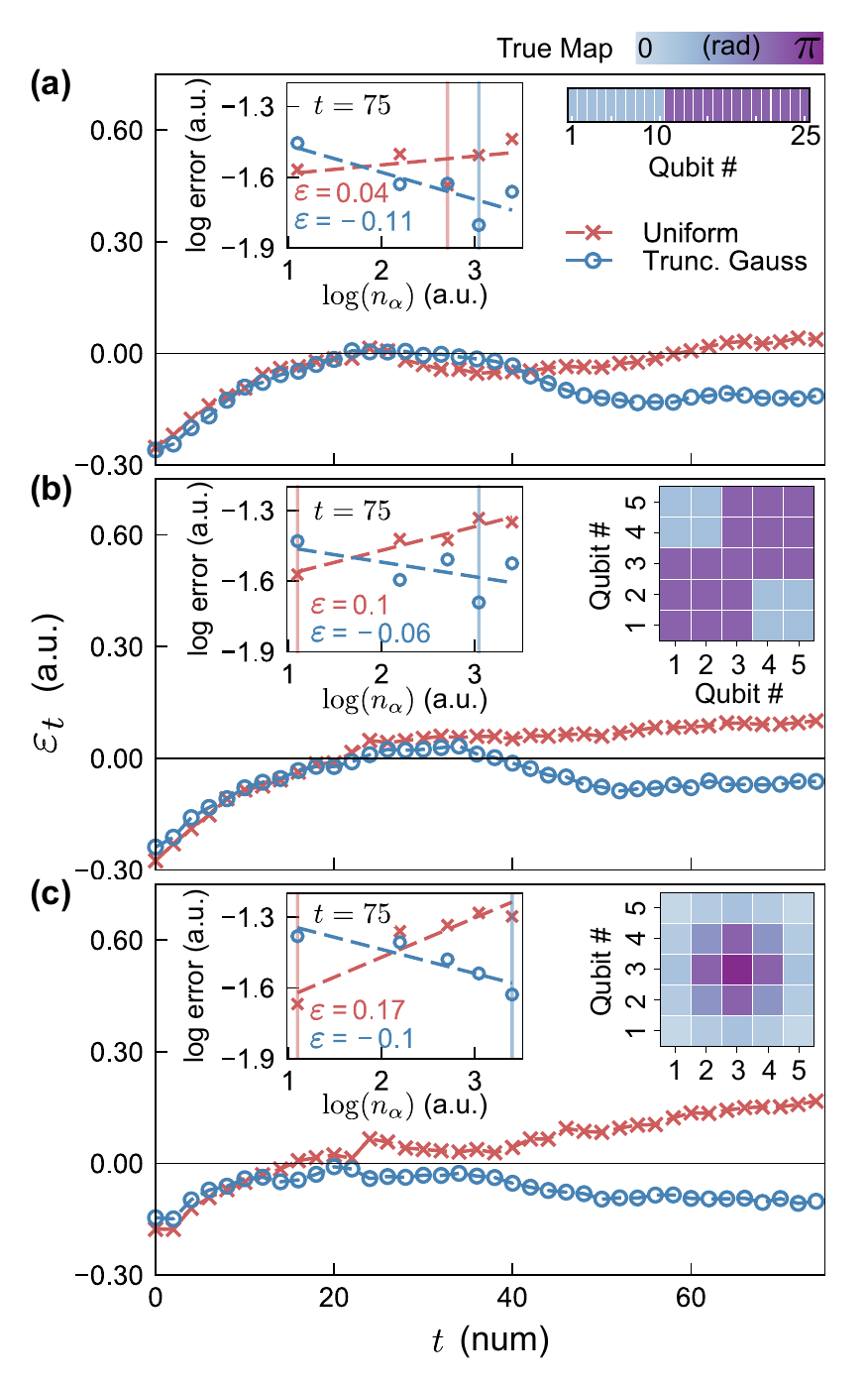}
	\caption[Results for NMQA error scaling behaviour for Uniform and Trunc. Gaussian.]{\label{fig_data_scaling} Error scaling behaviour for Uniform and Trunc. Gaussian. Rows represent 1D linear array, a 2D array with a square field, and a 2D array Gaussian field with $d=25$ (right insets); with high and low qubit phase values of $0.25\pi, 0.75\pi$ radians depicted on colorscales. (a)-(c) main panels depict $\varepsilon_t$ against $t$ for tuned parameters. $\varepsilon_t > 0$ for Uniform; $\varepsilon_t \in [-1, 0)$ for Trunc. Gaussian for $t \gg d$ agrees with typical convergence analysis. Data for Uniform (red crosses) and Trunc. Gaussian (blue circles). Left insets depict the log of the expected mean square map reconstruction error per qubit over 50 runs against the log of $n_\alpha$ number of $\alpha$-particles. From left to right, the $x$-axis shows increased particle number $n_\alpha = 3, 9, 15, 21, 30; n_\beta = \frac{2}{3} n_\alpha$; for $t=75$. $\varepsilon_t$ is the gradient of the line of best fit (dashed lines). Vertical colored lines mark particle configurations yielding lowest empirical error for tuned parameters ($\Sigma_v, \Sigma_F, \lambda_1, \lambda_2$) for Uniform: (a) $(6.0e^{-9}, 0.10, 0.88, 0.72)$; (b) $(7.1e^{-7}, 0.04, 0.88, 0.72)$ (c)$(5.9e^{-9}, 0.10, 0.72, 0.95)$. Trunc. Gaussian: (a) $(9.0e^{-8}, 2.6e^{-5}, 0.88, 0.72)$; (b) $(8.9e^{-7}, 1.9e^{-9},0.88, 0.72)$; (c) $(0.77, 4.6e^{-6}, 0.72, 0.95)$.}
\end{figure} 

Using this error metric for the three case studies examined in \cite{gupta2020adaptive}, a plot of the log true mean-square error per qubit against log number of particles yields the estimated slope $\varepsilon_t $ in 
\cref{fig_data_scaling}. For each case study, the true map $F^*$ over the arrangement of $d=25$ qubit locations is provided in the right insets in both 1D and 2D. In the main panel, we plot the extracted $\varepsilon_t$ against $t$ for both Uniform (red crosses) and Trunc. Gaussian (blue circles) expansion strategies. For each value of $t$, these $\varepsilon_t $ values are calculated from the gradient of a line of best fit for the log of true mean-square error per qubit in map reconstruction against the log of $n_\alpha$, the total number of $\alpha$ particles at the beginning and end of each $t$. Example raw and best fits for the case $t=75$ are shown in the left insets in (a)-(c). 

For $t\lesssim d=25$, we observe $\varepsilon_t < 0$ for both data sets. This means that increasing $n_\alpha$ under any expansion strategy (Uniform or Trunc. Gaussian) improves the inference procedure when data is sparse, assuming that the correct initial distribution has been specified. For the high-data regime, $t\gg d$, the values of $\varepsilon_t$ diverge between the two expansion strategies. The Uniform approach in (a)-(c) shows that $\varepsilon_t > 0$ as $t$ increases. In contrast, under a Trunc. Gaussian strategy, we see that $\varepsilon_t \in [-1, 0)$ is satisfied for all values of $t \gg d$ in all cases studied. 

These observations are consistent with our expectations. For the Uniform strategy, we expect $\varepsilon_t >0$ since filter convergence does not hold asymptotically as information about $R_t$ is reset to the prior distribution for $R_0$ even for large values of $t$. In this limit, increasing particle number $n_\alpha$ increases the level of randomness in the filtering distributions. By contrast, under a Trunc. Gaussian strategy, we expect that in some physical applications, it is reasonable to assume that that length-scale distributions are well described by the first two moments of an appropriately designed truncated Gaussian distribution at each $t$. Under these conditions, information transfer from $t$ to $t+1$ occurs such that filter convergence may hold and we expect that the condition $\varepsilon_t \in [-1, 0)$ is satisfied asymptotically. 

Thus, true error scales in a predictable way for \cref{algorithm:nmqa} under a Trunc. Gaussian approach. Our results provide compelling numerical evidence that classical convergence behaviour appears to hold even if single-qubit projective measurements and adaptive control features are incorporated into a classical filtering framework. Additional supporting numerical results are provided in \cref{app:bkg:3}. 

\section{\label{sec:Conclusion}Conclusion}

In this work, we explore a new implementation of adaptive filters for quantum systems with projective measurement  models  and  rigorously  demonstrate  that  the  theoretical  basis  of  classical  nonlinear  filtering  applies  in this context. Taking inspiration from classical signal processing, we combine discrete analysis of continuous amplitude signals with Born's rule and show that a novel likelihood function can be used to individually filter a sequence of single-shot projective measurements.  While this likelihood function can be incorporated in any classical filtering framework, we show that its inclusion into particle filtering methods preserves important convergence properties of particle-based solutions generalizable to a broad range of difficult inference problems encountered in quantum characterization and control.

Extending these insights, we investigate convergence of classical adaptive filtering of quantum projective measurements. These convergence properties are especially useful if practical implementations limit \emph{apriori} knowledge typically required for filter tuning or training machine learning methods. Indeed, the technical approach we introduce here for the modification of classical filtering algorithms is generalizable to a wide class of problems as we make minimal assumptions about measurement procedure, noise characteristics, or the dynamics of an open multi-qubit system. Applications include adaptive measurement selection \cite{gupta2020adaptive,gupta2020integration}, but other examples could include classical noise spectroscopy, efficient tomography, spatiotemporal forecasting, or adaptive calibration and control tasks using time-series of discrete projective measurements. 

Focusing on numerical studies for one such example in Ref.~\cite{gupta2020adaptive}, an empirical rate of convergence computed as the scaling factor, $\varepsilon_t$, of true error with particle number was shown to be theoretically expected to satisfy the condition $\varepsilon_t \in [-1, 0)$. This condition $\varepsilon_t \in [-1, 0)$ appears to hold for a range of algorithmic and physical configurations under a Trunc. Gaussian particle expansion strategy in a manner similar to convergence properties for conventional particle filtering. While these  numeric studies represent only one type of application of what is a broadly deployable algorithmic framework, the empirical results provide compelling evidence that it may be possible to extend conventional convergence theorems to our methods. 

Thus far we have put forth the idea that effect of quantum projective measurements on classical filtering methods can instead be understood as the effect of a discrete likelihood function on convergence properties of the underlying branching processes. In the case that these branching process can be viewed as classical random walks, for instance, in classification and regression tree analysis, the insights presented in this manuscript can be used to appropriately customize alternative stochastic frameworks for predictive-control. All of these stochastic methods have wide-ranging implications for device calibration, crosstalk analysis, non-Markovian noise characterization and automated system tuneup. We look forward to exploring how the rigorous analysis we have performed here may be applied to a broad class of adaptive filtering problems for near term quantum computers.

\section*{Data and Code Availability}
Unrestricted access to the codebase and data is provided via http://github.com/qcl-sydney/nmqa.\\

\section*{\label{sec:acknow}Acknowledgments}
R. Gupta would like to thank Andrew Doherty for extensive discussions. This work was partially supported by the US Army Research Office under Contract W911NF-12-R-0012, and a private grant from H. and A. Harley.


%

\clearpage

\appendix

\section{\label{app:bkg:0}Background to particle filters}

This Appendix is a primer on background theory for convergence analysis of  particle filters. Background concepts and key results from Ref. \cite{bain2009} are summarized to accompany \cref{sec:pf} of the main text, and derivations are provided to illustrate key stepping stones for the overall argument about the convergence of standard, bootstrap particle filters with multinomial particle branching mechanisms.

\subsection{\label{app:bkg:1} Convergence analysis of particle filters}

The key objective of any particle filter is to obtain an approximation to a posterior Bayesian distribution. Let $\measure{{}}{\cdot}$ denotes a probability measure, and let $\pi_t$ be the true posterior distribution in Bayesian analysis. Then $\pi_t$ is expressed as the conditional probability of $X_t$ given the \ofield{} generated by the observations $Y_{0:T}$,
\begin{align}
 \pi_t &:= \measure{{}}{X_t \in A | \ofieldgen{Y_{0:T}}}, \quad \forall A \in \mathcal{S}_X, \pi_t \in P(\mathbb{S}_X), \label{eqn:11}\\
 \pi_0 &= \mathcal{U}(\mathbb{S}_X). \label{eqn:12}
\end{align} In the above, the measure $\pi_t$ is a \textit{random} probability measure in the space of all possible measures $ P(\mathbb{S}_X)$. The space $\mathcal{S}_X$ is the Borel \ofield{} generated by the state-space $\mathbb{S}_{X}$ for all $t$. Similar comments apply to the \ofield{} generated by the observations $Y_{0:T}$. The notation $X_t \in A$ means that some instance of $X_t$ is an event $X_t \in A$ in the Borel \ofield{} generated by the state space of $X$. A list of useful definitions are summarized in \cref{table:notation}. The initial condition, $\pi_0$, is taken to be a uniform distribution over the state-space of $X_0$ for this manuscript.

If $\pi_t$ is obtained as a solution to a inference problem, it can be used to obtain information of functions of $X_t$,
\begin{align}
\pi_t f &= \ex{f(X_t) | \ofieldgen{Y_{0:T}}} \quad \forall f \in B(\mathbb{S}_X), A \in \mathcal{S}_X. \label{eqn:13}
\end{align} In the above, the term $B(\mathbb{S}_X)$ refers to a space of bounded, $\mathcal{S}_X$-measurable functions which correspond to transformations of the state (\emph{e.g.} dynamical evolution of $X_t$, measurement model for $X_t$). Here, $\ex{\cdot}$ refers to an expectation of a random variable or expectation of bounded, Borel-measurable functions of random variables, and provides the link between the solution of the particle filter and how information about $X_t$ can be obtained in analysis. 

In practical applications, a frequently used assumption in sequential Bayesian inference is that the process $X_{0:t}$ is a Markov chain, 
\begin{align}
 \measure{{}}{X_{t+1} \in A | \ofieldgen{X_{0:t}} } = \measure{{}}{X_{t+1} \in A | X_{t}}, \label{eqn:14}
 \end{align} where knowledge of the entire process $X_{0:t-1}$ can be safely discarded if $X_t$ is accessible. The transition kernel captures the probability $X_{t+1}$ occurs if the previous state was $X_t = x$,
\begin{align}
 K_t(x,A) & := \measure{{}}{X_{t+1} \in A | X_{t}=x}, \label{eqn:15}\\
 K_{t}(x, A) &: \mathbb{S}_X \times \mathcal{S}_X \to P(\mathbb{S}_X) \times B(\mathbb{S}_X), \label{eqn:16}\\
	& \quad \forall t = 0, 1, \hdots, A \in \mathcal{S}_X, x \in \mathbb{S}_X. \nonumber
\end{align} In the above, any transition kernel $K_{t}(x, A)$ for Markov chains satisfies the property that $ K_{t}(\cdot, A) \in B(\mathbb{S}_X)$ is a bounded Borel-measurable function for any $A \in \mathcal{S}_X$, and $ K_{t}(x, \cdot)$ is a probability measure over all possible final states at $t+1$ if $X_t = x$. The kernel $K_{t}(x, A)$ and the initial condition $X_0$ thus completely specify the statistical properties of the Markov chain $X_{0:t}$.

Specifically for sequential Bayesian analysis, many non-Markov classical random processes can be recast or are well approximated by Markov chains if $X_t$ and $K_{t}(x, A)$ are appropriated defined (\emph{e.g.} hidden Markov models, autoregressive moving average (ARMA) representations \cite{gupta2018machine}). The common feature of all these methods is that they assume properties of $K_{t}(x, A)$ are known \emph{a priori} or its parameters can be learned from data. This state-transition information is expressed in kernel notation, where $K_t\pi_t$ for any probability measure $\pi_t$ is shorthand for
\begin{align}
 K_t\pi_t (A) := \int_{\mathbb{S}_X} K_t(x,A) \pi_t(x) dx, \label{eqn:17}
\end{align} where the argument $A \in \mathcal{S}_X$ is a particular event under consideration, and measures are defined on $\mathbb{S}_X$. The appearance of $K_t(x,A) $ inside the integral aligns with the use of kernel nomenclature, and in the case where $t$ is a time index, $K_{t}(x, A)$ refers to the evolution of $X_t$. 

Having defined state transitions of $X_{t} \to X_{t+1} $, we now define the measurement process $X_{0:t} \to Y_{t} $. One typically defines a Markov measurement model, where the observation $Y_t$ depends only on the state $X_t$ and not the full process $X_{0:t}$. This measurement process is captured by a positive function, $g_t^{(Y_t=y_t)}(X_t)$ which satisfies, 
\begin{align}
 g_t^{(Y_t=y_t)}(X_t) &\in B(\mathbb{S}_X), \label{eqn:19}\\
 g_t^{(Y_t=y_t)}(X_t) & \geq 0, \quad \forall t. \label{eqn:20}
\end{align} In the above, the lowercase $y_t$ represents instances of the true random variate $Y_t$, and $B(\mathbb{S}_X)$ is the space of bounded Borel measurable function defined on $\mathbb{S}_X$. It is additionally assumed in this manuscript that $g_t^{(Y_t=y_t)}(X_t)$ is continuous with respect to $x \in \mathbb{S}_X$, and strictly positive, though this strictly positive condition can be relaxed \cite{bain2009} and it is not discussed in detail here. The notation $g_t := g_t^{(Y_t=y_t)}(X_t)$ is sometimes employed as shorthand in the remainder of this section. In the case that $g_t$ represents measurement noise density in conventional estimation theory, $g_t$ is called a likelihood function \cite{bain2009}. We adopt the nomenclature of a likelihood function to refer to $g_t$ for ease of reading, though materials of this section do not need to make this identification. 

Under these definitions, Bayes rule for the conditional probability of $X_t$ given observations $Y_{0:t}$ is written in recursive form as,
\begin{align}
	\pi_t &:= g_t * p_t (A) := \frac{\int_A g_t(x) dp_t(x) }{p_t g_t }, \label{eqn:22}\\
	p_t &:= K_{t-1} \pi_{t-1}, p_t \in P(\mathbb{S}_X), \label{eqn:23}\\
	p_t g_t &:= \int_{\mathcal{S}_X} g_t(x) dp_t(x) > 0 \label{eqn:24}
\end{align} The use of the projective product in the first line, $*$, is essentially a restatement of Bayes rule. In this product, the measure $p_t$ is a `predictive probability measure' in the sense that it uses the transition kernel for a Markov chain to obtain the distribution at $t$ if the distribution at $t-1$ is known. The integral in the numerator exists only over the outcomes $A$, whereas the denominator represents a normalization over all outcomes in the state space of $X$. The resulting product of $p_t$ and $g_t$ yields the true measure $\pi_t$ at $t$, and the next iteration is commenced for the incoming measurement $Y_{t+1}$.

The measures $\pi_t, p_t$ are true random measure revealed only through many experimental runs, and depend on the random observation record \emph{i.e.} $ \pi_t \equiv \pi_t^{Y_{0:t}}, p_t \equiv p_t^{Y_{0:t}}$. The notation above captures the concept that the observation record $Y_{0:t}$ is a random observation vector. In any single experimental run, we collect data by measuring instances of the true random process $Y_{0:t}$ and we obtain the `fixed' realization of the observation record, $y_{0:t} := \{Y_0 = y_0, Y_1 =y_1, \hdots, Y_t = y_t \}$. This means that once data, $y_{0:t}$, has been collected for a single experimental run, the measures $\pi_t^{y_{0:t}}$ and $p_t^{y_{0:t}}$ can be computed using this data. For any fixed path ($y_{0:t}$), these computations will in principle yield non-random quantities for $\pi_t^{y_{0:t}}$ and $p_t^{y_{0:t}}$. Thus, $\pi_t^{y_{0:t}}$ and $p_t^{y_{0:t}}$ should be distinguished from the random measures $\pi_t, p_t$. 

Under these definitions of the true random measures $\pi_t, p_t$, let $n$ represent the total number of particles $\{x_t^{(i)}\}_{i=1}^n$, each of which represents a hypothesis for $X_t$. Then the discrete approximation, $ \pi_t^n$, for the true $\pi_t$, is expressed as
\begin{align}
 \pi_t^n &:= \frac{1}{n}\sum_{i=1}^n \delta_{x_t^{(i)}}, \quad (\pi_t^n)_{n=1}^\infty \to \pi_t, \label{eqn:25}
\end{align} and similarly $ p_t^n = K_{t-1}\pi_{t-1}^n$ is the empirical sample for the true $p_t$ (cf. Remark 10.17 in \cite{bain2009}). In the above, the Kronecker delta, $\delta_{(\cdot)}$, is used because the approximate probability measures represent discrete probability distributions. 

Once a measurement result is received at $t$, a particle weight, denoted $G_t^{(i)}$, is computed for all $i=1,2, \hdots n$ and a weighted empirical distribution is formed. This weighted distribution, $\bar{\pi}_t$, is the set of both particles and their weights, expressed as
\begin{align}
 \bar{\pi}_t^n &:= \sum_{i=1}^n G_t^{(i)} \delta_{\bar{x}_t^{(i)}}, \quad \bar{x}_t^{(i)} \sim p_t^n. \label{eqn:26}
\end{align} Here, the particle weight, $G_t^{(i)}$ represents the probability of receiving a measurement $Y_t=y_t$ if the hypothesis captured by the $i$-th particle $X_t= \bar{x}_t^{(i)}$ is taken to be true. Raw weights for each particle are empirically normalized across $n$ particles for each $t$. The notation $\bar{x}_t^{(i)}$ indicates that $\bar{\pi}_t^n$ should be computed after evolving particles from $t-1$ into the current iteration at $t$, and $G_t^{(i)}$ are calculated based on a single measurement, $Y_t$, received at $t$. The additional randomness introduced by the particle approximations means that we have additional filtrations generated by the algorithm
\begin{align}
\mathcal{G}_t = \ofieldgen{x_s^{(i)}, \bar{x}_s^{(i)}, s \leq t, i = 1, \hdots, n} \label{eqn:27}\\
\bar{\mathcal{G}}_t = \ofieldgen{x_s^{(i)}, \bar{x}_s^{(i)}, \bar{x}_t^{(i)}, s < t, i = 1, \hdots, n}. \label{eqn:28}
\end{align} Here, $\bar{\mathcal{G}}_t \subset \mathcal{G}_t$, where $\bar{\mathcal{G}}_t $ includes particles $\bar{x}_t^{(i)}$ at start of iteration $t$ but excludes the posterior particles at the end of $t$. Thus, $\bar{\mathcal{G}}_t $ captures the output of a re-sampling step and it is used to describe the properties of the particle branching mechanism in \cref{prop:banchingproperties} (see below).

The arrow in \cref{eqn:25} invokes a notion of convergence, and the definition of convergence can take many forms. We focus on `convergence in expectation' and `almost sure' (a. s.) convergence \cite{bain2009}, which both imply convergence in probability and distribution \cite{evans2004probability}. For a sequence of random measures, $(\mu^n)_{n=1}^{\infty}$, and another random measure $\mu$, these are formalized as
\begin{align}
 \lim_{n \to \infty} \ex{|\mu^n f - \mu f|}& = 0, \quad \forall f \in C_b(\mathbb{S}), \label{eqn:29}\\
 \lim_{n \to \infty} \mu^n & = \mu, \quad \mathbb{P}\mathrm{-a.s.} \label{eqn:30} 
\end{align} The first line defines convergence in the expectation values for all continuous, bounded functions $f$. The second line defines a. s. convergence. 

Having established empirical distributions formed by discrete particles, and some notions of convergence, a general theorem for particle filters is restated below summarizing relevant known results reported in Ref. \cite{bain2009}.
\begin{theorem} [Bain \& Crisan, 2009] \label{bain:10.7}
 For all $f\in B(\mathbb{S}_X)$ and all $ t \in [0,T]$ the limits: \\
 (a0) $\lim_{n\to \infty} \ex{|\pi_t^{n,y_{0:t}} f - \pi_t f|} = 0$ \\
 (b0) $\lim_{n\to \infty} \ex{|p_t^{n,y_{0:t-1}} f - p_t f|} = 0 $\\
 hold if and only if for all $f\in B(\mathbb{S}_X)$ and all $t \in [0,T]$: \\
 (a1) $ \lim_{n \to \infty} \ex{|\pi_0^{n,y_{0:t}} f - \pi_0 f|} = 0$ \\
 (b1) $\lim_{n\to \infty} \ex{|p_t^{n,y_{0:t-1}} f - K_{t-1}\pi_{t-1} f|} = \lim_{n\to \infty} \ex{|\pi_t^{n,y_{0:t}} f - \bar{\pi}_t^{n,y_{0:t}}f|} = 0 $
\end{theorem} In the theorem above, $\pi_0^n$ is the particle approximation to the true prior, and the particles $x_0^{(i)} \sim \pi_0$ are sampled from the true initial distribution, $\pi_0$. The theorem above refers to a single fixed path, $Y_{0:t}=y_{0:t}$, constituting a single run of measurements. The same results hold for a random observation vector, $Y_{0:t}$ under the additional condition that there exists some constant $k_t$ such that $p_t g_t \geq k_t$. With these substitutions, the conditions (a1) and (b1) imply convergence in expectation of the sequences $\pi_t^{n,Y_{0:t}}$ ($p_t^{n,Y_{0:t-1}}$) to $\pi_t$ ($p_t$), and the proof is provided in the Appendix.

In order to use \cref{bain:10.7} to make necessary and sufficient statements about convergence (of any form), we restrict our discussion to the mapping of continuous bounded functions, $f$, and place additional requirements on the transition kernel and likelihood functions. The following theorem states these additional requirements and establishes the link with the branching mechanism for the particle filter, summarized and restated from Ref. \cite{bain2009}.
\begin{theorem} \label{bain:conv_summary}
Let $(\pi_t^{n,Y_{0:t}})_{n=1}^\infty$ and $(p_t^{n,Y_{0:t-1}})_{n=1}^\infty$ be measure valued sequences produced by particle approximations and branching mechanisms satisfying \cref{prop:banchingproperties} (below). Assume there exists some constant $k_t$ such that $p_t g_t \geq k_t$. Then:\\
(i) \cref{bain:10.7} holds for all $f\in B(\mathbb{S}_X)$ and all $ t \in [0,T]$ with the substitutions $\pi_t^{n,y_{0:t}} \to \pi_t^{n,Y_{0:t}}$ and $p_t^{n,y_{0:t-1}} \to p_t^{n,Y_{0:t-1}}$. \\
(ii) Assume that the transition kernel for $X$ is Feller and that the likelihood functions are all continuous for all $t\in [0,T]$. Then, the sequences $\pi_t^{n,Y_{0:t}}$ ($p_t^{n,Y_{0:t-1}}$) converges in expectation to $\pi_t$ ($p_t$) for all $t \in [0, T]$ if and only if conditions (a1) and (b1) are satisfied for all $f\in C_b(\mathbb{S}_X)$ and all $ t \in [0,T]$. \\
(iii) Assume that the transition kernel for $X$ is Feller and that the likelihood functions are all continuous for all $t\in [0,T]$. If the branching mechanism satisfying \cref{prop:banchingproperties} is a multinomial random process, then, $\lim_{n\to \infty} \pi_t^{n,Y_{0:t}} = \pi_t, \quad \mathbb{P}\mathrm{-a.s.}$ and $\lim_{n\to \infty} p_t^{n,Y_{0:t-1}} = p_t, \quad \mathbb{P}\mathrm{-a.s.}$
\end{theorem} 
In the above, part (i) implies convergence in expectation, but parts (ii) and (iii) are `if and only if' statements for our two notions of convergence. Of these, the requirement that the transition kernel $K_t$ is Feller \cite{bain2009,crisan2002survey} means that $K_t f \in C_b(\mathbb{S}_X)$ for all $f\in B(\mathbb{S}_X)$, \emph{i. e.} the transition kernel does not bring $f$ outside the space of continuous bounded functions. A particle filter is in-scope of \cref{bain:conv_summary} if the branching mechanism satisfy mathematical statements in \cref{prop:banchingproperties}.

\subsection{Supplementary proofs for convergence analysis for particle filters \label{app:bkg:2} }

The first step will be to establish the link between \cref{bain:10.7} and \cref{bain:conv_summary}. Of these, \cref{bain:10.7} is a re-statement of Theorem 10.7 in \cite{bain2009}. This theorem applies to a fixed observation path. In contrast, the proof below establishes \cref{bain:conv_summary} part (i) by extending \cref{bain:10.7} to conditions of convergence for \textit{random} measures. Here, the randomness of measures in \cref{bain:conv_summary} accrues from two sources (a) the random observation vector and (b) the particle approximation for continuous distributions. We will eventually find that the affect of these random sources is to require an additional condition on \cref{bain:conv_summary}. This condition is that for some non-zero constant $k_t$, the expected value of the product $p_tg_t \geq k_t$.

\begin{theorem}[Bain \& Crisan, 2009] \label{theorem:convergence}
	Assume that for any $t\geq 0$, there exists a constant $k_t$ such that $p_tg_t \geq k_t$. Then for all $f \in B(\mathcal{S}_X)$:
	\begin{align}
		\lim_{n\to\infty} \ex{ |\pi_t^n f - \pi_t f| } &= 0 \label{eqn:conv:1}\\
		\lim_{n\to\infty} \ex{ |p_t^n f - p_t f| } &= 0 \label{eqn:conv:2}
	\end{align} hold if and only if the following hold:
	\begin{align}
		\lim_{n\to\infty} \ex{ |\pi_0^n f - \pi_0 f| } &= 0 \label{eqn:conv:3} \\
		\lim_{n\to\infty} \ex{ |p_t^n f - K_{t-1}\pi^{n}_{t-1} f|} &= 0 \label{eqn:conv:4}\\
		\lim_{n\to\infty} \ex{ |\pi_t f - \bar{\pi}^n_t f|} &= 0 \label{eqn:conv:5}
	\end{align} where $\bar{\pi}^n_t := g_t * p_t^n$ and $\bar{\pi}^n_t f =p^n_t(fg_t)/p^n_t g_t $.
\end{theorem}
\begin{proof}
	Assume \cref{eqn:conv:1,eqn:conv:2} hold. Then \cref{eqn:conv:1} $\implies$ \cref{eqn:conv:3} by setting $t=0$. 	
	
	By the triangle inequality, for all $f \in B(\mathbb{S}_X)$,
	\begin{align}
		&|p_t^n f - K_{t-1}\pi^{n}_{t-1} f| \nonumber \\
		\leq & |p_t^n f - K_{t-1}\pi_{t-1} f| + | K_{t-1}\pi_{t-1} f - K_{t-1}\pi^{n}_{t-1} f|, \\
		\leq & |p_t^n f - p_t f| + | \pi_{t-1} (K_{t-1}f) - \pi^{n}_{t-1} (K_{t-1}f)|. \label{eqn:conv:extrareq:1}
	\end{align} Taking expectations of both sides as $n \to \infty$, the two terms on the right hand side are zero from \cref{eqn:conv:1,eqn:conv:2} and we recover \cref{eqn:conv:4}. 
	
For \cref{eqn:conv:5}, we consider the following expression,
	\begin{align}
		\pi_t f - \bar{\pi}^n_t f = & \frac{p_t(fg_t) }{p_t g_t} - \frac{p^n_t(fg_t)}{p^n_t g_t}, \\
		= &\frac{p_t(fg_t) }{p_t g_t} - \frac{p^n_t(fg_t)}{p^n_t g_t} + \frac{p^n_t(fg_t)}{p_t g_t} - \frac{p^n_t(fg_t)}{p_t g_t}, \\
		= & \frac{1}{p_t g_t} \left( p_t(fg_t) - p^n_t(fg_t)\right) + \frac{p^n_t(fg_t)}{p_t g_t} - \frac{p^n_t(fg_t)}{p^n_t g_t},\\
		= & \frac{1}{p_t g_t} \left( p_t(fg_t) - p^n_t(fg_t)\right) + \frac{p^n_t(fg_t)}{p_t g_tp^n_t g_t} \left(p^n_t g_t - p_t g_t \right).
	\end{align}
	Denote $\frac{p^n_t(fg_t)}{p^n_t g_t} \leq \parallel f_\infty \parallel$ and talking the absolute values and expectation on both sides yields:
	\begin{align}
		&\ex{|\pi_t f - \bar{\pi}^n_t f|} 
		\nonumber \\
		\leq & \ex{\frac{1}{p_t g_t} | p_t(fg_t) - p^n_t(fg_t)|} + \parallel f_\infty \parallel\ex{\frac{ 1}{p_t g_t} |p^n_t g_t - p_t g_t |}, \\
		\leq & \ex{\frac{1}{k_t} | p_t(fg_t) - p^n_t(fg_t)|} + \parallel f_\infty \parallel\ex{\frac{1}{k_t} |p^n_t g_t - p_t g_t |}, \\
		\leq & \frac{1}{k_t} \ex{| p_t(fg_t) - p^n_t(fg_t)|} +\frac{\parallel f_\infty \parallel}{k_t} \ex{ |p^n_t g_t - p_t g_t |}. \label{eqn:conv:extrareq:2}
	\end{align} Here, we invoked the assumption that there exists a constant $k_t$ such that $p_t g_t > k_t > 0$, such that $\frac{1}{k_t}$ is greater than $\frac{1}{p_tg_t}$, and this allows us to bring $k_t$ outside the expectation value.
	In the limit $n \to \infty$, both terms on the right hand side of the last line go to zero by \cref{eqn:conv:2}.
	
	In the reverse direction, we obtain \cref{eqn:conv:2} by applying the triangle inequality, and \cref{eqn:conv:1} follows by induction from using \cref{eqn:conv:3} as a starting point and \cref{eqn:conv:2,eqn:conv:4}. 

\end{proof} The derivation above establishes that \cref{bain:10.7} applies to random measures and \cref{bain:conv_summary} (i).

We now provide additional commentary around the requirements for the transition kernel and likelihood function which appear in both \cref{bain:conv_summary} (ii) and (iii) \cite{bain2009}. To obtain `if and only if' conditions for convergence in expectation using \cref{bain:10.7}, one needs to additionally argue that in proceeding to step \cref{eqn:conv:extrareq:1} that the transition kernel is Feller; and secondly, the limit in \cref{eqn:conv:extrareq:2} for both terms on the right-hand side is zero by invoking the continuity and boundedness of the likelihood function. This re-states Corollary 10.10 in \cite{bain2009} (fixed observation vector). In a similar manner, Corollary 10.30 in \cite{bain2009} (random observation vector) should be read as invoking these additional requirements on the likelihood function and transition kernel. These restrictions on the likelihood function and transition kernel enable us to interpret \cref{bain:conv_summary}(i) and \cref{bain:10.7} as an `if and only if' statement about convergence in expectation.

Next, the following theorem establishes conditions for almost-sure convergence of $p^n_t$ to $p_t$ and $\pi^n_t$ to $\pi_t$, encapsulated by \cref{bain:conv_summary} (iii).
\begin{theorem} [Bain \& Crisan, 2009] \label{theorem:convergence:almostsure}
	Assume that the transition kernel for $X$ is Feller, and the likelihood functions are continuous. Then the sequence $p^n_t$ converges to $p_t$ and $\pi^n_t$ converges to $\pi_t$ almost surely, for all $t \geq 0$, if and only if:
	\begin{align}
		& \lim_{n\to \infty} \pi_0^n = \pi_0 \quad \mathbb{P}-\mathrm{a.s.} \label{app:nmqa2:addin:1}\\
		& \lim_{n\to \infty} d(p^n_t, K_{t-1}\pi^n_{t-1}) = 0 \quad \mathbb{P}-\mathrm{a.s.} \label{app:nmqa2:addin:2}\\
		& \lim_{n\to \infty} d(\pi^n_t, \bar{\pi}^n_{t}) = 0 \quad \mathbb{P}-\mathrm{a.s.} \label{app:nmqa2:addin:3}
	\end{align} where $d(\cdot, \cdot)$ is any metric that generates a weak topology on the space of finite measures, and $\mathcal{M}$ is a convergence determining set in $C_b(\mathbb{S}_X)$. 
	\begin{proof}
		This is a re-statement of Theorem 10.12 (fixed observation vector) and Proposition 15 in \cite{bain2009} (random observation vector); proofs can be found within the reference. 
	\end{proof}
\end{theorem} In the above, the first condition in \cref{app:nmqa2:addin:1} establishes that we start from a good approximation of $\pi_0$. The remaining conditions state two requirements which must be met in satisfying the recursion relations for our approximations of the true posterior $\pi_t^n$. The condition in \cref{app:nmqa2:addin:2} states that the empirical predictive distribution and the transition kernel enable us to track the posterior from $t$ to $t+1$ `closely enough'. The condition in \cref{app:nmqa2:addin:3} states that empirical distributions before ($\bar{\pi}_t^n$) and after ($\pi_t^n$) the branching mechanism in the particle measure do not deviate within each iteration $t$. While this branching mechanism can be designed arbitrarily, the branching mechanism satisfying \cref{prop:banchingproperties} has desirable properties for convergence, including the identities below.

\begin{lemma} The following identities hold for the empirical distributions of a particle filter with the branching mechanism described in \cref{prop:banchingproperties}.
	\begin{align}
		\ex{p_t | \mathcal{G}_{t-1}} &= K_{t-1} \pi_{t-1}^n \label{eqn:identity:dynamics}\\
		\pi_t^n &= \frac{1}{n}\sum_{i=1}^n \xi_t^{(i)} \delta(\bar{x}_t^{(i)}) \label{eqn:rewrite_empirical_posterior_as_pibar}
	\end{align}
	\begin{proof}
		The first identity is obtained by substituting definitions for empirical measures as
		\begin{align}
			p_t^n &:= \frac{1}{n}\sum_{i=1}^n \delta(\bar{x}_t^{(i)}) = \frac{1}{n}\sum_{i=1}^n \delta(K_{t-1} x_{t-1}^{(i)} ), \label{eqn:identity:dynamics:substep}\\
			\implies \ex{p_t | \mathcal{G}_{t-1}} &= \frac{1}{n}\sum_{i=1}^n \ex{\delta(K_{t-1} x_{t-1}^{(i)} ) | \mathcal{G}_{t-1}}, \\
			&= \frac{1}{n}\sum_{i=1}^n K_{t-1} \delta( x_{t-1}^{(i)} ). \\
			&= K_{t-1} \frac{1}{n}\sum_{i=1}^n \delta( x_{t-1}^{(i)} ), \\
			&= K_{t-1} \pi_{t-1}^n.
		\end{align}
		For the second identity, each particle $\bar{x}_t^{(i)}$ replaces itself $\xi_t^{(i)}$ number of times, that is, $\bar{x}_t^{(i)}$ is counted up $\xi_t^{(i)}$ number of times in the re-sampled posterior $\pi_t^n$. Hence we can write the re-sampled posterior as
		\begin{align}
			\pi_t^n &\equiv \sum_{i=1}^{n} \mathbb{P}(\bar{x}_t^{(i)} \mathrm{chosen})	\cdot \delta(\bar{x}_t^{(i)}), \\
			&= \sum_{i=1}^{n} \left(\frac{\xi_t^{(i)}}{\sum_{k=1}^n {\xi_t^{(k)}}} \right)	\delta(\bar{x}_t^{(i)}), \\
			&= \sum_{i=1}^{n} \left(\frac{\xi_t^{(i)}}{n} \right)	\delta(\bar{x}_t^{(i)}) .
		\end{align} The last line follows from the previously stated assumption that the number of particles are held constant after each re-sampling step. 
	\end{proof}
\end{lemma} 

We now link the branching mechanism of the particle filter to the main convergence theorems listed previously. 
\begin{theorem} [Bain \& Crisan, 2009] \label{bain:pfconvergence}
	Assume that for all $t$, there exists a constant $k_t$ such that $p_tg_t \geq k_t$ and the covariance matrix of a branching mechanism satisfies $q^t A_t^n q \leq n c_t$ . Then for a random observation vector $\pi_t^n \equiv \pi_t^{n, Y_{0:t}}$ and $p_t^n \equiv p_t^{n, Y_{0:t}}$ yield
	\begin{align}
		\lim_{n\to \infty} \ex{|\pi_t^n f - \pi_t f|} &= 0, \\
		\lim_{n\to \infty} \ex{|p_t^n f - p_t f|} &= 0,
	\end{align}
	for all $f \in B(\mathbb{S}_X), t \geq 0$. 
	
	\begin{proof}
		Following the proof structure suggested in \cite{bain2009}, we use \cref{theorem:convergence}. First, we observe show convergence for the initial conditions. Since $\delta(x_0^{(i)}) \sim \pi_0$ and $\pi_0^n(1) = 1, \forall n$ , then the dominated convergence theorem for measure valued random variables applies
		\begin{align}
			\lim_{n \to \infty} \pi_0^n = \pi_0 \implies \lim_{n \to \infty} \ex{|\pi_0^n f - \pi_0f|} = 0, \quad \mathbb{P}-\mathrm{a.s.}
		\end{align} This verifies \cref{eqn:conv:3}.
		For \cref{eqn:conv:4}, we use \cref{eqn:identity:dynamics}, $\ex{p_t f | \mathcal{G}_{t-1}} = \pi_{t-1}^n (K_{t-1}f) \equiv \mu$,
		\begin{align}
			&\ex{(p^n_tf - \pi_{t-1}^n (K_{t-1}f))^2 | \mathcal{G}_{t-1}} \nonumber \\
            =&  \ex{(p^n_tf)^2 | \mathcal{G}_{t-1}} - \mu^2, \\
			=& \ex{(p^n_tf)^2 | \mathcal{G}_{t-1}} - \pi_{t-1}^n (K_{t-1}f)^2, \\
			=& \frac{1}{n^2} \ex{\sum_{i=1}^n \sum_{j=1}^n f(\bar{x}_t^{(i)}) f(\bar{x}_t^{(j)}) | \mathcal{G}_{t-1}} - \pi_{t-1}^n (K_{t-1}f)^2 \label{bain:pfconvergence:step:1},\\
			=& \frac{1}{n^2} \sum_{i=1}^n \ex{(f(\bar{x}_t^{(i)}))^2 | \mathcal{G}_{t-1}} - \pi_{t-1}^n (K_{t-1}f)^2 \label{bain:pfconvergence:step:2},\\
			=& \frac{1}{n} \pi_{t-1}^n \left(K_{t-1} f^2-(K_{t-1}f)^2\right) \label{bain:pfconvergence:step:3},\\
			&\implies \ex{(p^n_tf - \pi_{t-1}^n (K_{t-1}f))^2} \leq \frac{\parallel f \parallel_\infty^2}{n}.
		\end{align} In going from \cref{bain:pfconvergence:step:1} to \cref{bain:pfconvergence:step:2}, the sums in the first term are brought outside of the expectation value due to independence of samples $\bar{x}_t^{(i)}$ conditional on $ \mathcal{G}_{t-1}$, invoking both linearity of expectation values and independence of the predictive samples $f(\bar{x}_t^{(\cdot)})$ at $t$ conditional on $ \mathcal{G}_{t-1}$. The first term in \cref{bain:pfconvergence:step:3} follows by applying the same reasoning to \cref{eqn:identity:dynamics:substep}. Finally, in the last step, we assumed that $f$ is a bounded, Borel measurable function \cite{szekeres2004course}, and there exists $ f_\infty$, a the limiting vector for all $t, n$ such that $\pi^n_t f \leq \parallel f \parallel_\infty, \forall n$ \cite{bain2009}.
		
		The limit in \cref{eqn:conv:4} is implied by the last line since for any random variable, $x$, $0 \leq \mathbb{V}\{x\} \leq \ex{x^2}$. Almost sure convergence of the second moment and variance of $x$ to zero implies $\ex{x}^2$ almost surely; letting $x := p^n_tf - \pi_{t-1}^n (K_{t-1}f)$ implies \cref{eqn:conv:4}.
		
		For \cref{eqn:conv:5}, one requires that
	 \begin{align}
		\ex{|\pi_t f - \bar{\pi}^n_t f|} \leq \ex{|\pi_t f - \pi^n_t f|} + \ex{|\pi^n_t f - \bar{\pi}^n_t f|}.
	\end{align} In the above, the first term on the right hand side goes to zero by \cref{eqn:conv:1}. Hence, we need only to show $\ex{|\pi^n_t f - \bar{\pi}^n_t f|} \to 0$ as $n \to \infty$. We derive an analogous result for $\ex{(\pi_t^nf - \bar{\pi}_t^nf)^2}$, by re-writing key quantities in terms of $n$ dimensional vectors whose elements represent individual particles. In the equations below, the probability of choosing particle $\bar{x}_t^{(i)}$ is proportional to its weight $G_t^{(i)}$, $W_t$ is the vector of all particle weights, and $\xi_t$ is a vector containing frequencies of particles after re-sampling. These definitions are provided below, 
		\begin{align}
			(\pi_t^nf - \bar{\pi}_t^nf)^2 &:= (W_t f_t)^T (W_t f_t), \label{app:nmqa2:defadd:1} \\
			f_t &:= f(\delta(\bar{x_t})), \label{app:nmqa2:defadd:2}\\
			\delta(\bar{x_t}) & = [\delta(\bar{x_t}^{(1)}), \hdots, \delta(\bar{x_t}^{(n)})]^T, \label{app:nmqa2:defadd:3}\\		
			W_t &:= \frac{\xi_t - nw_t}{n}, \label{app:nmqa2:defadd:4}\\
			\xi_t &:= [ \xi_t^{(1)}, \hdots, \xi_t^{(n)}], \label{app:nmqa2:defadd:5}\\
			G_t &:= [G_t^{(1)}, \hdots, G_t^{(n)}] \label{app:nmqa2:defadd:6},
		\end{align} while the additional quantities $c_t$ and $A^n_t$ are defined by Condition 3 in \cref{prop:banchingproperties}. Further since $f$ is a bounded, Borel measurable function, we assume there exists $ f_\infty \geq f(\delta(\bar{x_t})) $ i.e. the limiting vector for all $t, n$. Using the above definitions, 
		\begin{align}
			& \ex{(\pi_t^nf - \bar{\pi}_t^nf)^2} \nonumber \\
            & = \ex{(W_t f_t)^T (W_t f_t)} \label{app:nmqa2:addin:4},\\
			& = \frac{1}{n^2}\ex{f_t^T (\xi_t - nw_t)^T(\xi_t - nw_t) f_t } \label{app:nmqa2:addin:5}, \\
			& \leq \frac{1}{n^2}\ex{f_\infty^T (\xi_t - nw_t)^T(\xi_t - nw_t) f_\infty } \label{app:nmqa2:addin:6}, \\
			& = \frac{1}{n^2} f_\infty^T \ex{ (\xi_t - nw_t)^T(\xi_t - nw_t) } f_\infty \label{app:nmqa2:addin:7}, \\
			& = \frac{1}{n^2} f_\infty^T A_t^n f_\infty \label{app:nmqa2:addin:8}, \\
			& = \frac{\parallel f\parallel_\infty^2}{n^2} \frac{f_\infty^T}{\parallel f\parallel_\infty} A_t^n \frac{f_\infty}{\parallel f\parallel_\infty} \label{app:nmqa2:addin:9},\\
			& = \frac{\parallel f\parallel_\infty^2}{n^2} q^T A_t^n q, \quad q = \frac{f_\infty}{\parallel f\parallel_\infty} \label{app:nmqa2:addin:10},\\
			& \leq \frac{\parallel f\parallel_\infty^2}{n^2} c_t n	\label{app:nmqa2:addin:11}, \\
			&\implies \ex{(\pi_t^nf - \bar{\pi}_t^nf)^2}  \leq \frac{c_t \parallel f\parallel_\infty^2}{n} \label{eqn:convergence:lemma:pibar:squared}.
		\end{align} In the above, \cref{eqn:rewrite_empirical_posterior_as_pibar} enables one to rewrite the left hand side in terms of $W_t$ in \cref{app:nmqa2:addin:4}. These expressions involving $W_t$ are expanded out subsequent steps, where $f$ is replaced by its limiting vector in \cref{app:nmqa2:addin:6}, and this nonrandom quantity is brought outside the expectation value in \cref{app:nmqa2:addin:7}. The definition of $A^n_t$ in Condition 3 of \cref{prop:banchingproperties} is used in \cref{app:nmqa2:addin:8}. These terms are arranged in \cref{app:nmqa2:addin:9,app:nmqa2:addin:10} by defining a nonrandom quantity $q$. The inequality in Condition 3 of \cref{prop:banchingproperties} is used to obtain \cref{app:nmqa2:addin:11}, leading to $\frac{1}{n}$ dependence in the final step.
	\end{proof}
\end{theorem}

\cref{bain:pfconvergence} establishes the link between branching mechanisms in \cref{prop:banchingproperties} and \cref{bain:conv_summary} in the main text.

We now focus on branching mechanisms that are multinomial random processes. In this special case, \cref{prop:banchingproperties} are satisfied and the constant $c_t$ can be given a value to further specify the rate of convergence. Both of these results can be found in \cite{bain2009} and they are listed for completeness below. 

\begin{lemma} [Bain \& Crisan, 2009]
	If the offspring distributions are multinomial, then all properties of the branching mechanism are satisfied. 
	\begin{proof}
		Let $\xi_t^{(i)}$ in \cref{app:nmqa2:defadd:5} be a multinomially distributed where the probability of choosing particle $\bar{x}_t^{(i)}$ is proportional to its weight $G_t^{(i)}$ in \cref{app:nmqa2:defadd:6}, 
		\begin{align}
			\mathbb{P}(\xi_t^{(i)} = m^{(i)}) := \frac{n!}{\prod_{i=1}^{k} m^{(i)}!} \prod_{i=1}^{k}( G_t^{(i)})^{m^{(i)}}.
		\end{align} Then by the properties of the multinomial distribution, we obtain
		\begin{align}
			\ex{\xi_t^{(i)} | \bar{\mathcal{G}}_t } &= n G_t^{(i)}, \\
			(A_t^n)_{i,j} &=\begin{cases}
				n G_t^{(i)} (1 - G_t^{(i)}), \quad i =j \\
				- n G_t^{(i)} G_t^{(j)}, \quad i \neq j,
			\end{cases}
		\end{align} where definitions in \cref{prop:banchingproperties,app:nmqa2:defadd:5,app:nmqa2:defadd:6} have been used. This implies
		\begin{align}
			q^T A_t^n q_t & := \sum_{i=1}^n \sum_{j=1}^n q_i q_j (A_t^n)_{i,j}, \\
			&= \sum_{i=1}^n q_i q_i n G_t^{(i)} (1 - G_t^{(i)}) - \sum_{i\neq j}^n \sum_{j=1}^n n q_i q_j G_t^{(i)} G_t^{(j)}, \\
			&= \sum_{i=1}^n q_i q_i n G_t^{(i)} - \sum_{i= 1}^n \sum_{j=1}^n n q_i q_j G_t^{(i)} G_t^{(j)} ,\\
			&= n \sum_{i=1}^n (q_i)^2 G_t^{(i)} - n (\sum_{i= 1}^n q_i G_t^{(i)})^2 ,\\
			&\leq n \sum_{i=1}^n (q_i)^2 G_t^{(i)} ,\\
			&\leq n \sum_{i=1}^n G_t^{(i)}, \quad |q_i| \leq 1 \forall i,\\
			&= n,\\
			\implies c_t &= 1 \quad \forall t.
		\end{align} 
	\end{proof}
\end{lemma}
The results thus far establish \cref{bain:conv_summary} (ii).
\begin{corollary}
	If the off-spring distributions are multinomial and under the conditions of \cref{theorem:convergence:almostsure} for the random observation vector, then 
	\begin{align}
		\lim_{n \to \infty} & p_t^{n} = p_t, \quad \mathbb{P}-\mathrm{a.s.} \\
		\lim_{n \to \infty} & \pi_t^{n} = \pi_t, \quad \mathbb{P}-\mathrm{a.s.}
	\end{align}
	\begin{proof}
		This is a partial restatement of Corollary 10.31 in \cite{bain2009}. It follows from \cref{theorem:convergence:almostsure} and \cref{bain:pfconvergence}.
	\end{proof}
\end{corollary}
The corollary above additionally establishes almost-sure convergence for \cref{bain:conv_summary} part (iii).

\section{Detailed adaptive filtering equations \label{app:bkg:3}}

In addition to the main text and Supplementary Materials in \cite{gupta2020adaptive}, we restate technical details of an adaptive filtering framework using projective measurements. The equations in \cref{sec:nmqa} are now supplemented by additional technical details below. 

Following on from \cref{eqn:33,eqn:34,eqn:35} in the main text, a single measurement $Y_t^{(j)}$ is received pertaining to the label $j$ and iteration $t$. Using this new information, we now seek to estimate the elements of $X$ pertaining to locations $q_t$ in the neighbourhood of $j$. Letting $\nu^{(j,q_t)}$ be the separation distance between two points labeled by $j$ and $q_t$,  the quantity $\chi_t^{(j,q_t)}$ is the estimate of the field at $q_t$ due to the measurement at $j$, 
\begin{align}
\chi_t^{(j,q_t)} := (1 - \lambda_2^{\tau^{(q_t)}}) F_t^{(q_t)} + \lambda_2^{\tau^{(q_t)}} F_t^{(j)} \exp\left( - \frac{(\nu^{(j,q_t)})^2}{2( R_t^{(j)})^2} \right), \label{eqn:38} 
\end{align} for some $\lambda_2^{\tau^{(q_t)}} \in [0,1]$ (see below). The set of relevant points, $q_t$, which define the neighbourhood about $j$, is expressed as
\begin{align}
Q_t^{(j)} &:= \{ q_t | \nu^{(j,q_t)} < k_0 R_t^{(j)} \}, k_0 \geq 1, \label{eqn:36}\\
\forall q_t &\in \{1, 2, \hdots d\} \setminus \{j\} \label{eqn:37},
\end{align} where $k_0$ is an arbitrary fixed number greater than unity, and $Q_t^{(j)}$ grows or shrinks even if the same value of $j$ is revisited during filtering. 

In this manner, both the size of the neighborhood $Q_t^{(j)}$ generated about $j$ and the approximate field estimated at neighboring points $q_t$ is parameterized by a single value, $R_t^{(j)}$. Further, the estimated state uncertainty pertaining to $R_t^{(j)}$ at the end of each iteration $t$ is used to adaptively select a control action for the next iteration $t+1$, as discussed below.  

\subsection{Adaptive control via $\beta_\alpha$-particle layer} 

The empirical variance of $\beta_\alpha$-particles can be used to compute the expected value of a Fano factor $C_t^{(j)}$ for each label $j$,
\begin{align}
	C_t^{(j)} = \mathbb{E}_{\alpha} \left[ \frac{1}{r_t^{(j, \alpha)}} \mathbb{V}_{\beta_\alpha}\left[{ x_t^{(j, \alpha, \beta_\alpha)}}\right] \right].\label{eqn:44}
\end{align} In the above, the empirical variance is computed with respect to the $\beta_\alpha$-particles for each $\alpha$ and the expected Fano factor is computed using the empirical mean of the $\alpha$-particles. The resulting quantity $C_t^{(k)}$ for $k=1, 2, \hdots d$ qubit locations is used in the control step, where the next physical measurement at $t+1$ is scheduled at the coordinate point for the label $j'$ associated with maximal uncertainty for map gradients,
\begin{align}
	j'=\mathrm{argmax}_k \{C_t^{(k)}\}_{k=1}^d. \label{eqn:45}
\end{align} 

The high level algorithmic structure is summarized by \cref{fig:NMQA_detailed} and outlined as pseudocode in \cref{algorithm:nmqa}, combining particle filtering (steps (i)-(x)), a control decision (step (xi)), and information-sharing (step (xii)). As with all particle filters, the efficacy of our adaptive particle filtering algorithm depends on both particle weight calculations and branching mechanisms, discussed below.

\begin{figure*}[t!]
 \centering
	\includegraphics[width=\textwidth]{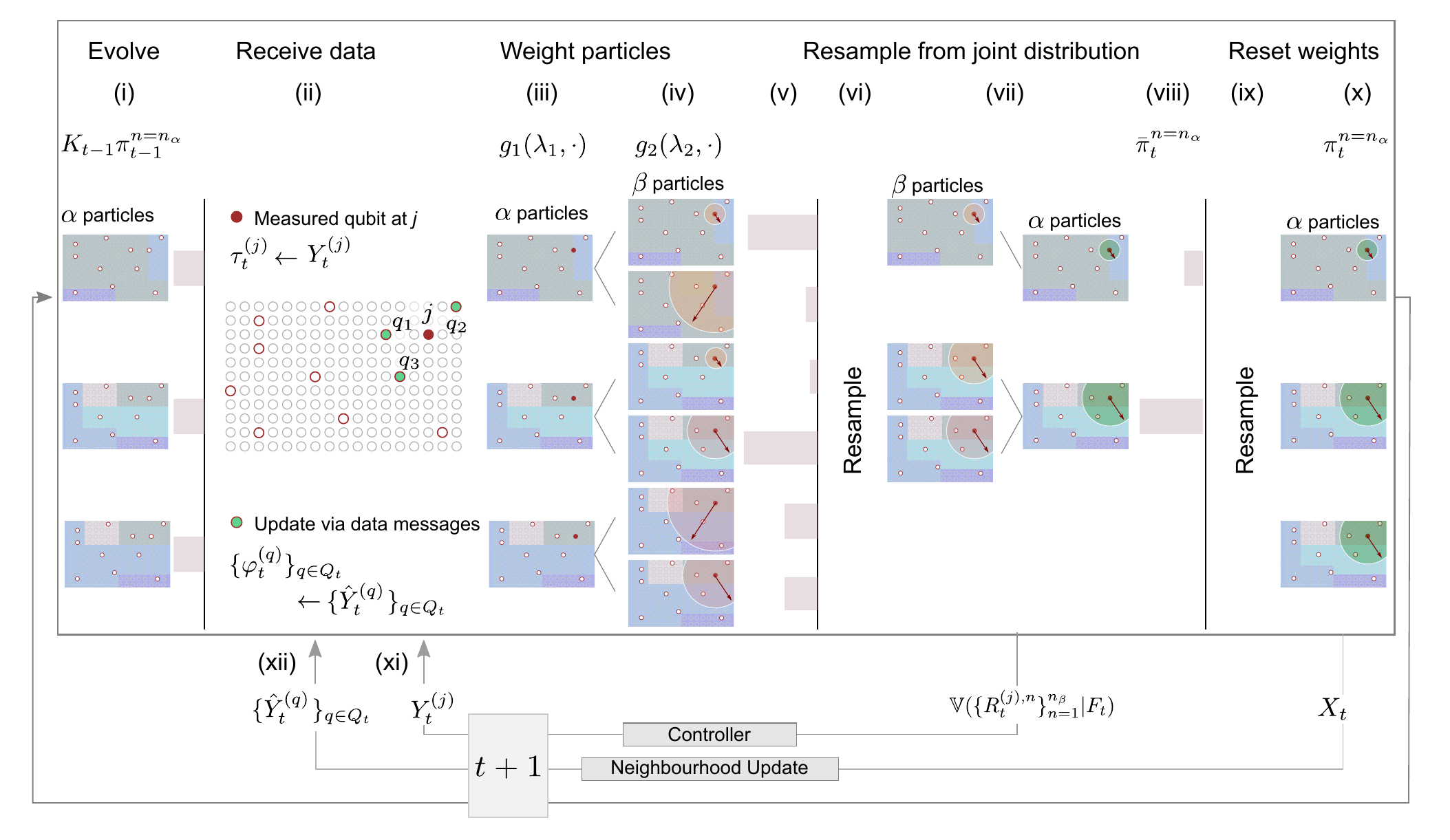}
	\caption[Schematic overview of adaptive particle filter at $t$.]{\label{fig:NMQA_detailed} Particle filter under an iterative likelihood approximation at $t$. $\alpha$-particles carry map information and $\beta$-particles represent length-scale information, $R_t^{(j)}$, at a measured location $j$; $n_\alpha = 3$, $n_\beta = 2$. (i) The posterior from $t-1$ is dynamically evolved to $t$ via $K_{t-1}$ yielding $p_t^{n_\alpha}$. (ii) A physical measurement at $j$ for iteration $t$. Shared data messages from $t-1$ are updated via $\varphi_t^{(q)}$ for $q\in Q_{t-1}$ the posterior neighborhood at $t-1$. (iii)-(iv) Proposed likelihood functions score individual particles. (iii) Each $\alpha$-particle is scored using $g_1^{(y_{t}^{(j)}, \alpha)}(\lambda_1, \Lambda_t^{(j, \alpha)}) $. (iv) Assuming $F_{t}$ for each $\alpha$-particle is known, a $\beta$-particle is scored via $g_2^{(j, \alpha, \beta_\alpha)}(\lambda_2, \Lambda_t)$. (v) The global scores $g_1^{(y_{t}^{(j)}, \alpha)}(\lambda_1, \Lambda_t^{(j, \alpha)}) g_2^{(j, \alpha, \beta_\alpha)}(\lambda_2, \Lambda_t)$ are computed for $n_\alpha n_\beta$ particles. (vi) Particles are re-sampled, resulting in new offspring particles depicted as leaves in (vii). (vii) The empirical mean and variance estimates of $R_t^{(j)}$ are stored for the $\beta$-layer of each $\alpha$ particle; the posterior estimate is passed onto the controller in (xi). (viii) The $\beta$-layer is collapsed for each $\alpha$-particle giving rise to the weighted empirical distribution, $\bar{\pi}_t^{n_\alpha}$. (ix) A second re-sampling step occurs with $n_\alpha$ particles using the empirical distribution $\bar{\pi}_t^{n_\alpha}$ in (viii). (x) Output of (ix) results in the final posterior distribution $\pi_t^{n_\alpha}$ at $t$; the mean of this distribution reflects the best information about $X_t$ conditioned on data, and this is shared with neighboring qubits in the posterior neighborhood in (xii). (xi) Controller picks the location for the next measurement where estimated variance metric is highest. (xii) Data messages are generated for the neighbors of $j$ in the posterior neighborhood using \cref{eqn:40}.}
\end{figure*}

\subsection{Particle weight calculations}

 We outline how $(\alpha,\beta_\alpha)$ particle-pairs are scored using a global scoring function. This function incorporates both physical single-qubit measurements and shared data-messages, and we distinguish it from the conventional usage of the word likelihood function in Bayesian analysis that is associated with the density of measurement noise only. If a measurement is performed at $j$ and received at iteration $t$, then the global scoring function for each particle pair is given by
\begin{align}
	g_{t}^{(y_{t}^{(j)}, \alpha, \beta_\alpha)}(\lambda_1, \lambda_2, \Lambda_t) := g_1^{(y_{t}^{(j)}, \alpha)}(\lambda_1, \Lambda_t^{(j, \alpha)}) g_2^{(j, \alpha, \beta_\alpha)}(\lambda_2, \Lambda_t). \label{eqn:46}
\end{align} In the above, the argument $ \Lambda_t^{(k, \alpha)}$ is a set of parameters involving calculations using physical measurements $Y_t^{(k)}$ and shared data messages $\hat{Y}_t^{(k)}$ at some qubit location $k\in \{1,2,\hdots d \}$. The global scoring function is composed of a product of two functions. Using the form in \cref{eqn:7,eqn:8}, the function $g_1^{(j, \alpha)}(\lambda_1, \Lambda_t^{(j, \alpha)})$ scores $\alpha$-particles following a measurement at $j$,
\begin{align}
	g_1^{(y_{t}^{(j)}, \alpha)}(\lambda_1, \Lambda_t^{(j, \alpha)}) & = \frac{\rhoq{0}}{2} \nonumber \\
	& + \rhoq{0} s(h(\lambda_1, \Lambda_t^{(j, \alpha)})) \left( \delta_{y_{t}^{(j)},1} - \delta_{y_{t}^{(j)},0}\right), \label{eqn:47}
\end{align} where $y_t^{(j)}$ is a instance of the observed physical measurement $Y_t^{(j)}$. The term $h(\cdot)$ is related to estimated information about $F_t^{(j)}$ carried by the parent $\alpha$-particle and it is discussed in detail below. Meanwhile, the function $g_2^{(j, \alpha, \beta_\alpha)}(\lambda_2, \Lambda_t)$ scores $\beta_\alpha$-particles for each $\alpha$ parent,
\begin{align}
	& g_2^{(j, \alpha, \beta_\alpha)}(\lambda_2, \Lambda_t) \nonumber \\
   &= \prod_{q_t \in Q_t^{(j)}} \frac{1}{k_1} \exp \left( \frac{-(h(\lambda_1, \Lambda_t^{(q_t, \alpha)}) - \chi_t^{(j, q_t, \alpha, \beta_\alpha)} - \mu_F)^2}{2\Sigma_F} \right). \label{eqn:48}
\end{align} The equation above represents a comparison of the best knowledge of the map at all neighbours with the value implied by smearing map in the neighborhood candidate, $r_t^{(j, \alpha, \beta_\alpha)}$. Here, the errors from approximating a continuously varying $F_t$ with a collection of overlapping Gaussians is assumed to be a truncated Gaussian with a mean $\mu_F$ and variance $\Sigma_F$. An integration constant $k_1$ arises from the state-space $\mathbb{S}_F$ which is assumed to be a finite interval on the positive real line in \cref{eqn:35}. The term $\chi_t^{(j, q_t, \alpha, \beta_\alpha)}$ is \cref{eqn:38} rewritten explicitly for particle calculations rather than the posterior estimates of the particle distribution,
\begin{align}
	\chi_t^{(j, q_t, \alpha, \beta_\alpha)} = & (1 - \lambda_2^{\tau_t^{(q_t)}}) h(\lambda_1, \Lambda_t^{(q_t, \alpha)}) \nonumber \\
 &+ \lambda_2^{\tau_t^{(q_t)}} h(\lambda_1, \Lambda_t^{(j, \alpha)}) \exp \left(\frac{(-\nu^{(j, q_t)})^2}{(r_t^{(j, \alpha, \beta_\alpha)})^2}\right). \label{eqn:49}
\end{align} In the equation above, as well as previous scoring calculations of this subsection, the term $h(\lambda_1, \Lambda_t^{(\cdot, \alpha)})$ features prominently and we now discuss this term. This data association mechanism $h(\lambda_1, \Lambda_t^{(\cdot, \alpha)})$ is a Markov function of physical measurements and data-messages at $t$ and it is given by the inverse $s^{-1}(\cdot)$, 
\begin{align}
	h(\lambda_1, \Lambda_t^{(k, \alpha)}) &:= s^{-1}\left(H(\lambda_1, \Lambda_t^{(k, \alpha)}) - \frac{1}{2}\right). \label{eqn:50}
\end{align} Here, $ H(\lambda_1, \Lambda_t^{(k, \alpha)})$ represents the process of extracting features from raw data to update the map hypothesis associated with the particle. The quantity $H(\lambda_1, \Lambda_t^{(k, \alpha)})$ are the empirical calculations for physical measurements ($ \kappa_t^{(j, \alpha)}$) and data messages ($\gamma_t^{(j, \alpha)}$) generated during the filtering procedure,
\begin{align}
H(\lambda_1, \Lambda_t^{(k, \alpha)}) & := \begin{cases} \left(1 - \frac{\lambda_1^{\tau_t^{(k)}}}{2}\right)\kappa_t^{(k, \alpha)} + \frac{\lambda_1^{\tau_t^{(k)}}}{2} \gamma_t^{(k, \alpha)}, \label{eqn:51} \\
	\quad \mathrm{for} \quad \tau_t^{(k)},\varphi_t^{(k)} > 0, \\
	\kappa_t^{(k, \alpha)}, \quad \tau_t^{(k)} \neq 0, \varphi_t^{(k)} = 0, \\
	\gamma_t^{(k, \alpha)}, \quad \tau_t^{(k)} = 0, \varphi_t^{(k)} \neq 0, \\
	\kappa_0^{(k, \alpha)}, \quad \tau_t^{(k)} = 0, \varphi_t^{(k)} = 0,
\end{cases}
\end{align} where $\Lambda_t^{(j, \alpha)}$ is the set of parameters,
\begin{align}
	\Lambda_t^{(j, \alpha)} &:= \{\tau_t^{(j)}, \varphi_t^{(j)}, \kappa_t^{(j, \alpha)}, \gamma_t^{(j, \alpha)}\}, \label{eqn:52}\\
	\Lambda_t: &= \{ \{ \Lambda_t^{(j, \alpha)} \}_{\alpha=1}^{\alpha=n_\alpha} \}_{j=1}^{j=d}. \label{eqn:53}
\end{align} In the above, $\tau_t^{(j)}$ ($ \varphi_t^{(j)}$) are counts of the number of measurements (data-messages) received at $j$. For a physical measurement at $j$ and data-message at $q_t$ for any labeled point $k = 1, 2, \hdots d $, the counting parameters in the set $\Lambda_t^{(k, \alpha)}$ are updated recursively, 
\begin{align}
	\tau_t^{(k)} &= \tau_{t-1}^{(k)} + \delta_{k,j}, && \tau_0^{(k)} =0, \label{eqn:54}\\
	\varphi_t^{(k)} &= \varphi_{t-1}^{(k)} + \delta_{k,q_t}, && \varphi_0^{(k)} =0. \label{eqn:55}
\end{align} Similarly, $ \kappa_t^{(j, \alpha)}$ ($\gamma_t^{(j, \alpha)}$) are Markov calculations which depend on the measurements (data-messages) at $j$. These calculations also depend on the particle index $\alpha$ via the initial particle sample. For general mapping problems where dynamics are known \emph{a priori} via the transition kernel $K_t$, then it is postulated that some learnable weights $ \{\{ \zeta_{t'}^{(k)} \}_{k=1}^d \}_{t'=0}^t $ and $ \{\{ \varsigma_{t'}^{(k)} \}_{k=1}^d \}_{t'=0}^t $ exist such that,
\begin{align}
\kappa_t^{(k, \alpha)} & = \kappa_{t-1}^{(k, \alpha)} + \zeta_t^{(k)} y_t^{(j)} \delta_{k,j} \label{eqn:56},\\
\gamma_t^{(k, \alpha)} & = \gamma_{t-1}^{(k, \alpha)} + \varsigma_t^{(k)} \hat{y}_t^{(j)} \delta_{k,j} \label{eqn:57},\\
\kappa_0^{(k, \alpha)} &= \mathcal{Q}(\frac{1}{2} + s(f_0^{(k, \alpha)})),\quad \gamma_0^{(k, \alpha)} = \kappa_0^{(k, \alpha)}. \label{eqn:58}
\end{align} Here the appropriate form of the weights are known \emph{a priori} or are assumed discoverable through additional dynamical inference, for example, using recursive least squares or Kalman filtering approaches. 



Our framework has free model parameters $\lambda_1$ and $\lambda_2$, as well as noise parameters $\Sigma_v, \mu_F, \Sigma_F$. For the case $\lambda_1 = \lambda_2 = 0$, no information sharing occurs in our framework (see \cref{eqn:49,eqn:51}). This means that the data-association via $h(\cdot)$, the generation of $\beta_\alpha$ particles, and the control step have no impact on the overall inference procedure. In the opposite limit that $\lambda_1 \to 1$ and $\lambda_2 \to 1$, our framework treats data-messages on an equal footing with physical measurements. For non-zero $\lambda_1$ and $ \lambda_2$, our framework departs from the explicit scope of \cref{bain:conv_summary}. 

\begin{figure}[t]
	\includegraphics[scale=1.0]{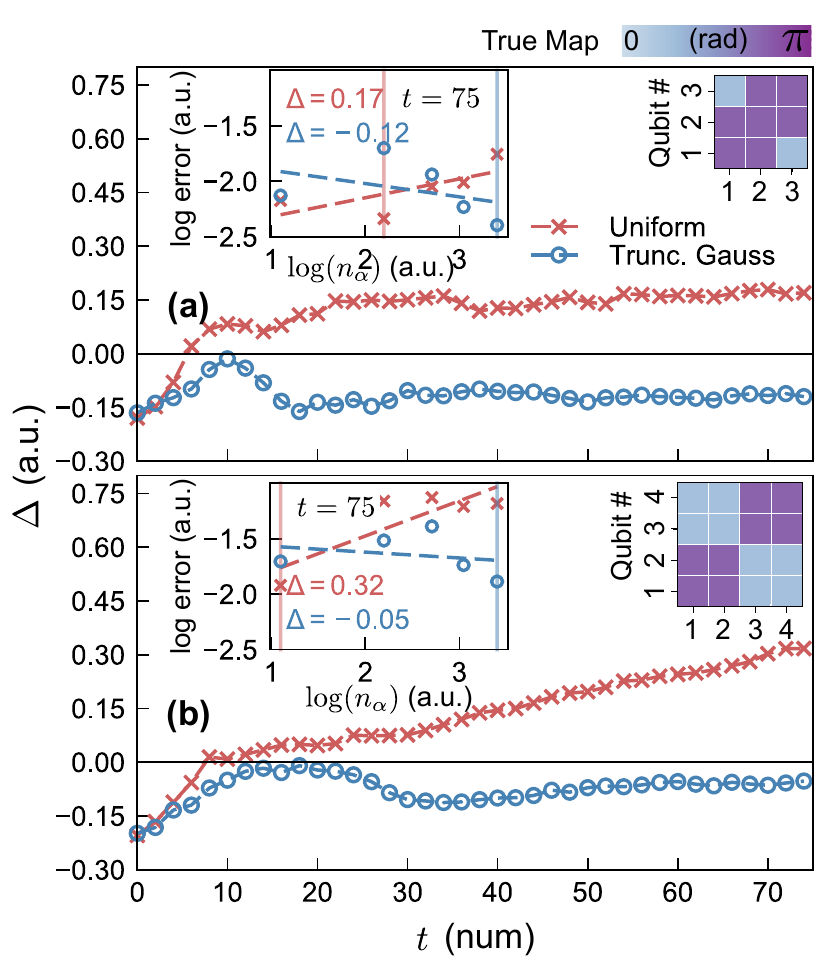}
	\caption{\label{fig_appendix_scaling} Error scaling behaviour for Uniform and Trunc. Gaussian. Rows represent increasing qubit system size $d=9, 16$ for 2D Square Field (right insets); with high and low qubit phase values of $0.25\pi, 0.75\pi$ radians depicted on heat-maps. (a)-(b) main panels depict $\Delta$ against $t$ for tuned NMQA parameters. $\Delta > 0$ for Uniform; $\Delta \in [-1, 0)$ for Trunc. Gaussian for $t \gg d$ agrees with \cref{bain:conv_summary}. Data for Uniform (crosses) and Trunc. Gaussian  (open circles). Left insets depict the log of the expected mean square map reconstruction error per qubit over 50 runs against the log of $n_\alpha$ number of $\alpha$-particles. From left to right, the $x$-axis shows increased particle number $n_\alpha = 3, 9, 15, 21, 30; n_\beta =  \frac{2}{3} n_\alpha$; for $t=75$. $\Delta$ is the gradient of the line of best fit (dashed lines). Vertical colored lines mark particle configurations yielding lowest empirical error for tuned parameters ($\Sigma_v, \Sigma_F, \lambda_1, \lambda_2$) for Uniform: (a) $(7.1e^{-7}, 0.05, 0.93, 0.68)$; (b) $(4.2e^{-3}, 2.6e^{-4}, 0.88, 0.72)$. Trunc. Gaussian: (a) $(6.3e^{-7}, 7.9e^{-7}, 0.95, 0.84)$; (b) $(4.2e^{-3}, 2.6e^{-4}, 0.93, 0.68)$.}
\end{figure}

\begin{figure}[t]
	\includegraphics[scale=1.0]{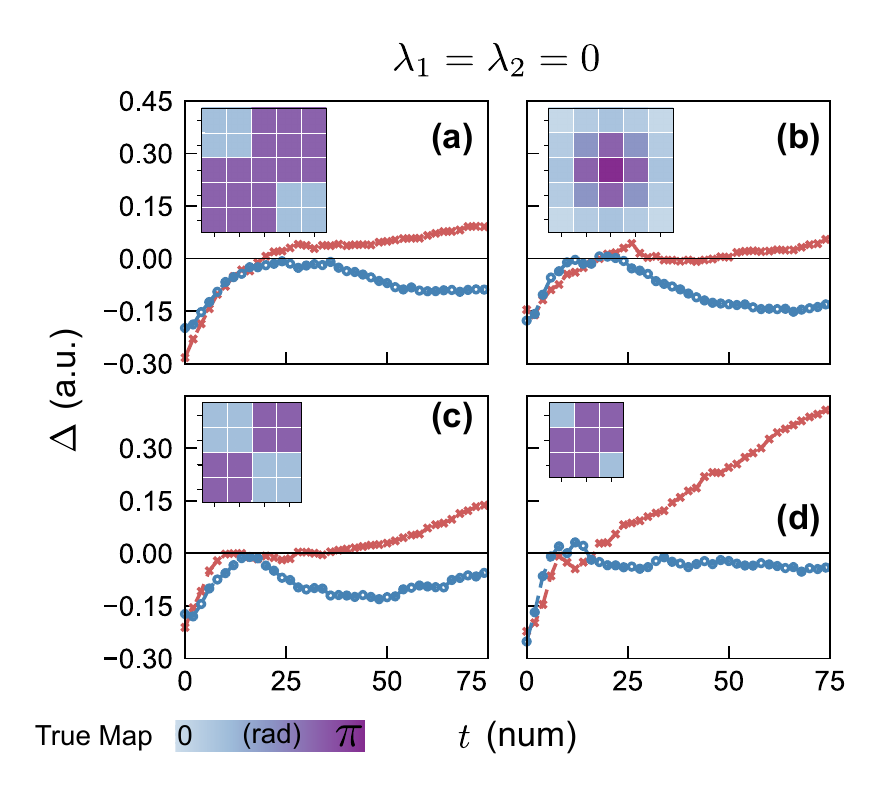}
	\caption{\label{fig:appendix:zerolambda} Error scaling behaviour for $\lambda_1=\lambda_2=0$. Panels represent different physical configurations (left insets): (a), (c), (d) 2D Square field for $d=25, 16, 9$ respectively; (b) 2D Gaussian field for $d=25$. (a)-(d) main panels depict $\Delta$ against $t$ for tuned NMQA parameters. Data for Uniform (crosses) and Trunc. Gaussian  (open circles); $\Delta > 0$ for Uniform; $\Delta \in [-1, 0)$ for Trunc. Gaussian for $t \gg d$ agrees with \cref{bain:conv_summary}. Tuned variance parameters ($\Sigma_v, \Sigma_F$) for Uniform: (a) $(7.1e^{-7}, 0.047)$; (b) $(5.9e^{-9}, 0.096)$; (c) $(4.2e^{-3}, 2.6e^{-4})$; (d) $(7.1e^{-7}, 0.047)$. Trunc. Gaussian: (a) $(8.9e^{-7}, 1.9e^{-9})$; (b) $(0.77, 4.6e^{-6})$; (c) $(4.2e^{-3}, 2.6e^{-4})$; (d) $(6.3e^{-7}, 7.9e^{-7})$. $n_\alpha = 3, 9, 15, 21, 30; n_\beta =  \frac{2}{3} n_\alpha$. High and low true field values of $0.25\pi, 0.75\pi$ radians.}
\end{figure} 

\subsection{Additional numerical analysis}

For the numerical analysis in the main text, $ \kappa_t^{(j, \alpha)}$ and $\gamma_t^{(j, \alpha)}$ are interpreted as the recursive forms of the empirical means of the physical measurements and data messages if the map $F_t$ is static in $t$. This yields uniform weights
\begin{align}
	\zeta_{t'}^{(k)} &= \frac{1}{\tau_t^{(k)}}, && \forall t' \in 0, 1, \hdots t, \label{eqn:73}\\
	\varsigma_{t'}^{(k)} &= \frac{1}{\varphi_t^{(k)}}, && \forall t' \in 0, 1, \hdots t. \label{eqn:74}
\end{align} To complete the specification of state spaces and the likelihood functions, substituting $F_{\min}=0, F_{\max}=\pi$ in \cref{eqn:35} gives the explicit form of $k_1$ in \cref{eqn:48},
\begin{align}
	k_1 = \frac{1}{2}\left(\erf(\frac{\pi + \mu_F}{\sqrt{2\Sigma_F}}) + \erf(\frac{\pi - \mu_F}{\sqrt{2\Sigma_F}})\right). \label{eqn:72}
\end{align} Meanwhile, the state-space of $R_t$ is set by size of physical hardware, where $R_{\min}$ is the minimal pairwise qubit separation and $R_{\max}$ is set to be some multiple of the maximum pairwise qubit separation on the device, in units of distance. 

We will present additional numerical evidence to support the conclusions of the main text. First our analysis focuses on the error scaling behaviour for NMQA with system size in 2D for a Square field. In \cref{fig_appendix_scaling}, we plot the error scaling factor  for particle number, $\Delta$, against number of iterations, $t$, for the two system sizes $d=9, 16$, where $d=25$ corresponds to \cref{fig_data_scaling}(b) in the main text. The behaviour of $\Delta$ also depicts features at $t \approx d$ suggesting two different regimes of behaviour apply, i.e.  $t<d$ (sparse data) and $t \gg d$ (high data limit). The data for Uniform (Trunc. Gaussian) is given in red crosses (blue circles). Our results agree with our theoretical expectations set by \cref{bain:conv_summary}: in the high $t$ regime, $\Delta >0$ for Uniform and the condition $\Delta \in [-1, 0)$  is satisfied for Trunc. Gaussian. 
 \FloatBarrier
In \cref{fig:appendix:zerolambda}, we provide scaling behaviour information for the case $\lambda_1=\lambda_2=0$ for all cases discussed in the main text and this appendix. Here, we have turned off the sharing mechanism in NMQA by setting $\lambda_1, \lambda_2$ to zero; where this choice leads to larger true expected mean square error per qubit than picking non-zero $\lambda_1, \lambda_2$. Nevertheless, despite the sub-optimal choice of setting $\lambda_1=\lambda_2=0$, this is not a difficult regime to analyse theoretically. We expect  \cref{bain:conv_summary} to hold under Trunc. Gaussian for $\lambda_1=\lambda_2=0$. Numerically, we confirm that  $\Delta \in [-1, 0)$  is indeed satisfied for Trunc. Gaussian for $t > d$ in a variety of physical 2D configurations in  \cref{fig:appendix:zerolambda}(a)-(d). The expectation that $\Delta >0$ for  Uniform for $t>d$ is additionally satisfied. Similar results hold in 1D. 

Collectively, these numerical results provide additional support for the conclusions of the main text.

\end{document}